\documentclass[12pt]{article}
\parskip=\medskipamount

\usepackage{verbatim}
\usepackage{tensor}
\usepackage{amsmath,amssymb,calc, amsthm,bbm, epsfig,psfrag, mathtools}
\usepackage{latexsym,bm,amsfonts}

\usepackage{slashed}

\newtheorem{thm}{Theorem}[section]

\newtheorem{lemma}{Lemma}[section]
\newtheorem{corollary}{Corollary}[section]

\usepackage{graphicx, enumerate}
\newcommand{\ul}[1]{{\underline{#1}}}
 \usepackage[all,cmtip]{xy}
\usepackage[numbers,sort&compress]{natbib}
\usepackage[dvipsnames]{xcolor}
\usepackage{mathrsfs} 
\usepackage{graphicx,tikz,tikz-cd}
\usepackage{quiver}
\usetikzlibrary{shadows}
\usetikzlibrary{shapes.arrows}

\usepackage{braket}
\usepackage{enumitem}
\usepackage{physics}

\allowdisplaybreaks

\newcommand{\Comment}[1]{{}}
\definecolor{darkblue}{rgb}{0.15,0.35,0.55}
\definecolor{reddish}{rgb}{0.65, 0.2, 0.2}
\usepackage[linktocpage=true]{hyperref}
\hypersetup{
colorlinks=true,
citecolor=darkblue,
linkcolor=reddish,
urlcolor=darkblue,
pdfauthor={},
pdftitle={},
pdfsubject={}
}

\newcommand{\im}{\mathrm{im}}

\makeatletter
\renewcommand\section{\@startsection {section}{1}{\z@}%
                                   {-3.5ex \@plus -1ex \@minus -.2ex}%nn
                                   {2.3ex \@plus.2ex}%
                                   {\normalfont\large\bfseries}}
\renewcommand\subsection{\@startsection{subsection}{2}{\z@}%
                                     {-3.25ex\@plus -1ex \@minus -.2ex}%
                                     {1.5ex \@plus .2ex}%
                                     {\normalfont\bfseries}}
\makeatother

\newcommand{\wotimes}{ \, \widehat{\otimes} \, }

\newcommand{\walpha}{ { \widetilde{\alpha} } }
\newcommand{\wbeta}{ { \widetilde{\beta} } }
\newcommand{\wgamma}{ { \widetilde{\gamma} } }
\newcommand{\wdelta}{ { \widetilde{\delta} } }

\newfont{\goth}{ygoth.tfm scaled 1200}                   % gothic font (usual)
 \numberwithin{equation}{section}

\newcommand{\overbar}[1]{\mkern 1.5mu\overline{\mkern-1.5mu#1\mkern-1.5mu}\mkern 1.5mu}

\newcommand{\Honeton}{{\mathcal{H}_{\ul{1} \ldots \ul{n}}}}
\newcommand{\soneton}{{s_{\ul{1} \ldots \ul{n}}}}
\newcommand{\rhooneton}{{\rho_{\ul{1} \ldots \ul{n}}}}

\newcommand{\psioneton}{{\psi_{\ul{1} \ldots \ul{n}}}}

\begin{document}
%%%%%%%%%%%%%%%%
%%%%%%%%%%%%%%%%
\begin{titlepage}
\begin{flushright}
\today
\end{flushright}
\vspace{5mm}

\begin{center}
{\Large \bf 
Hodge Theory for Entanglement Cohomology
}
\end{center}

\begin{center}

{\bf
Christian Ferko${}^{a,b}$,
Eashan Iyer${}^c$, 
Kasra Mossayebi${}^{d}$,
and
Gregor Sanfey${}^{e}$
} \\
\vspace{5mm}

\footnotesize{
${}^{a}$
{\it 
Department of Physics, Northeastern University, Boston, MA 02115, USA
}
 \\~\\
${}^{b}$
{\it 
The NSF Institute for Artificial Intelligence
and Fundamental Interactions
}
 \\~\\
 ${}^{c}$
{\it 
Brown University, Providence, RI 02912, USA
}
 \\~\\
${}^{d}$
{\it 
Monarch Park Collegiate Institute, Toronto, Canada}
}
 \\~\\
${}^{e}$
{\it 
The Latymer School, London N9 9TN, United Kingdom
}
\vspace{3mm}
~\\
\texttt{c.ferko@northeastern.edu, 
eashan\_iyer@brown.edu, 
everestkas@gmail.com,
y18grsan@latymer.co.uk
}\\
\vspace{3mm}

\end{center}

\begin{abstract}
\baselineskip=14pt

\noindent We explore and extend the application of homological algebra to describe quantum entanglement, initiated in \cite{Mainiero:2019enr}, focusing on the Hodge-theoretic structure of entanglement cohomology in finite-dimensional quantum systems. We construct analogues of the Hodge star operator, inner product, codifferential, and Laplacian for entanglement $k$-forms. We also prove that such $k$-forms obey versions of the Hodge isomorphism theorem and Hodge decomposition, and that they exhibit Hodge duality. As a corollary, we conclude that the dimensions of the $k$-th and $(n-k)$-th cohomologies coincide for entanglement in $n$-partite pure states, which explains a symmetry property (``Poincar\'e duality'') of the associated Poincar\'e polynomials.

\end{abstract}
\vspace{5mm}

\vfill
\end{titlepage}

\newpage
\renewcommand{\thefootnote}{\arabic{footnote}}
\setcounter{footnote}{0}

\tableofcontents{}
\vspace{1cm}
\bigskip\hrule

%%%%%%%%%%%%%%%%%%%%%%%%%%%%%%%%%%%%%%%%%%%%%%%%%%%%%%
%%%%%%%%%%%%%%%%%%%%%%%%%%%%%%%%%%%%%%%%%%%%%%%%%%%%%%
%%%%%%%%%%%%%%%%%%%%%%%%%%%%%%%%%%%%%%%%%%%%%%%%%%%%%%

\allowdisplaybreaks

\section{Introduction}

One of the important open challenges in quantum information (QI) is to characterize and understand the possible patterns of entanglement in systems consisting of arbitrary tensor products of finite dimensional Hilbert spaces. Although this problem is of great interest for both QI and related fields like theoretical physics and pure mathematics, it appears to be fearsomely complicated in the general case. Indeed, it has even been suggested that this problem may be, in some sense, equivalent to the study of quantum gravity itself.\footnote{An accessible exposition of this idea, along with useful references, can be found in \cite{Wall:2023myf}.}

In this work, we will retreat from the study of entanglement in quantum gravity and quantum field theory, and focus instead on patterns of entanglement in finite-dimensional quantum mechanical systems. In one sense, this represents a substantial restriction in scope, since our analysis will be applicable only to non-relativistic quantum systems and not to systems with infinitely many local degrees of freedom. However, considering finite quantum systems will offer us other advantages, such as the ability to use the machinery of finite-dimensional vector spaces, and the option to generate examples of entangled states numerically and search for additional patterns using a computer.

It is well-known that, in the arena of finite-dimensional quantum systems, the case of bipartite entanglement is much simpler to address than general multipartite entanglement. In fact, for a Hilbert space $\mathcal{H}_{AB} = \mathcal{H}_A \otimes \mathcal{H}_B$ consisting of two tensor product factors, a complete characterization of the entanglement in a state $\ket{\psi_{AB}} \in \mathcal{H}_{AB}$ is offered by the von Neumann entropy of the reduced density matrix. That is, given
\begin{align}
    \rho_{AB} = \ket{\psi_{AB}} \bra{\psi_{AB}} \, , \qquad \rho_A = \tr_B \left( \rho_{AB} \right) \, ,
\end{align}
where $\tr_B$ is the partial trace\footnote{The definition of the partial trace, and of other objects introduced here, will be reviewed shortly.} over the Hilbert space $\mathcal{H}_B$, the quantity
\begin{align}
    S_A = - \tr \left( \rho_A \log \left( \rho_A \right) \right) \, ,
\end{align}
contains all of the data about entanglement between subsystems $A$ and $B$ in $\ket{\psi_{AB}}$. 

However, even in the case of a tripartite Hilbert space, the patterns of possible entanglement are richer, and cannot be straightforwardly reduced to a single numerical quantity which measures the degree of entanglement. A famous illustration of this fact comes from the observation that there are two qualitatively different ``flavors'' of entangled tripartite systems of qubits \cite{PhysRevA.62.062314}, one represented by the GHZ state
\begin{align}\label{GHZ}
    \ket{ \text{GHZ} } = \frac{1}{\sqrt{2}} \left( \ket{000} + \ket{111} \right) \, ,
\end{align}
and one exemplified by the W state,
\begin{align}\label{W}
    \ket{W} = \frac{1}{\sqrt{3}} \left( \ket{001} + \ket{010} + \ket{100} \right) \, ,
\end{align}
and these two types of entangled states are physically inequivalent, in the sense that one cannot be transformed into the other by local quantum operations. Of course, the number and character of the different types of inequivalent entangled states only becomes more intricate as one increases the number and dimension of the subsystems.

It seems likely that new mathematical tools will be required in order to understand the possible patterns of entanglement in general $n$-partite systems. In this work, we will adopt the philosophy of \cite{Mainiero:2019enr}, which advocates for the view that cohomology is a natural framework in which to understand quantum entanglement. Such cohomological techniques are ubiquitous in physics, from the study of classical gauge theory (where a non-trivial de Rham cohomology of the spacetime manifold signals the existence of flat field configurations which are not pure gauge), to BRST quantization (where physical states are identified with cohomology classes associated with an operator $Q_B$), to supersymmetric quantum mechanics (where cohomology classes count ground states in the model). Likewise, the strategy of \cite{Mainiero:2019enr} is to construct a certain cochain complex associated with a given quantum state $\ket{\psi}$ in a multipartite Hilbert space, with the property that the cohomologies of this complex capture data about aspects of entanglement in the state $\ket{\psi}$. This approach is similar in spirit to that of the earlier work \cite{e17053253}, which introduced a notion of \emph{information (co)homology} (see \cite{Vigneaux2017InformationSA} for a complementary discussion and further developments).

An analogy may be helpful to illustrate, heuristically, why cohomology could be an appropriate gadget for this task. In the more familiar setting of de Rham cohomology, the Poincar\'e lemma guarantees that any closed differential form is \emph{locally} exact (i.e. it is exact on a sufficiently small open set). The cohomology of a manifold measures whether one is obstructed from promoting this local exactness to a statement of \emph{global} exactness, by joining together the presentations as an exact form on the various small open patches. Let us view two Hilbert spaces $\mathcal{H}_A$ and $\mathcal{H}_B$ as being in analogy with two small open sets. Any quantum state $\ket{\psi_A} \in \mathcal{H}_A$ or $\ket{\psi_B} \in \mathcal{H}_B$ in a Hilbert space with a single tensor product factor is trivially a product, just as any closed form on a small open set is exact. However, there is an obstruction to realizing a ``global'' state $\ket{\psi_{AB}} \in \mathcal{H}_A \otimes \mathcal{H}_B$ as a product expression $\ket{\psi_A} \otimes \ket{\psi_B}$. It stands to reason that such an obstruction may be characterized by some type of cohomology. But the impossibility of realizing a state $\ket{\psi_{AB}}$ as a tensor product $\ket{\psi_A} \otimes \ket{\psi_B}$ is precisely the statement that this state is entangled.\footnote{A similar analogy can be used to motivate a homological approach to \emph{classical} probability, where one is (in general) obstructed from realizing a probability distribution on several variables as a product of univariate distributions. The application of homology to this problem has been studied in \cite{10.1007/978-3-031-38271-0_25}.
}

In \cite{Mainiero:2019enr}, this intuition was made precise via the definitions of two cochain complexes, which are referred to as the \texttt{GNS} and \texttt{Com} complexes in that work, associated with any given state $\ket{\psi}$.\footnote{See also \cite{Harvey:2020jvu}, where these tools were applied to the study of a state $\ket{0_L}$ which arises in the context of quantum error correcting codes, and \cite{Hamilton:2023klg}, where a different entanglement complex was constructed.} In the present article, we will be chiefly concerned with \texttt{Com}, the commutant complex. Our primary reason for focusing on the latter of these two complexes is that, as we will show, the commutant complex admits a natural analogue of Hodge theory. This allows one to apply various results and intuition associated with the study of harmonic forms on manifolds. For instance, in the case of de Rham cohomology on a closed and oriented $n$-dimensional manifold $\mathcal{M}$, the Hodge theorem gives an isomorphism between the $k$-th cohomology group $H^k ( \mathcal{M} )$ and the space $\mathrm{Harm}^k ( \mathcal{M} )$ of harmonic $k$-forms:
\begin{align}
    H^k ( \mathcal{M} ) \cong \mathrm{Harm}^k ( \mathcal{M} ) \, .
\end{align}
Furthermore, the Hodge star operation $\ast$ gives an isomorphism between the space $\mathrm{Harm}^k ( \mathcal{M} )$ of harmonic $k$-forms and the space $\mathrm{Harm}^{n-k} ( \mathcal{M} )$ of harmonic $(n-k)$ forms:
\begin{align}\label{cong_iso}
    \mathrm{Harm}^k ( \mathcal{M} ) \cong \mathrm{Harm}^{n-k} ( \mathcal{M} ) \, .
\end{align}
Combining these two isomorphisms, one immediately recovers the result that the dimensions of the $k$-th cohomology group and the $(n-k)$-th cohomology group coincide. When phrased in terms of the Betti numbers $b_k$, which measure the dimensions of homology groups $H_k ( \mathcal{M} )$ that are themselves related to the dimensions of cohomology groups by Poincar\'e duality, this is an avatar of the familiar symmetry property $b_k = b_{n-k}$.

The development of analogous Hodge-theoretic observations for entanglement in finite quantum systems will be one of the main focuses of this work. Since a version of the isomorphism (\ref{cong_iso}) exists for the commutant complex \texttt{Com} but not for the GNS complex \texttt{GNS}, we focus on the former, which we will simply call the entanglement complex. Likewise, we refer to the cohomology of this cochain complex as entanglement cohomology.

The structure of this paper is as follows. Section \ref{sec:ent_coho_review} reviews the basic notions of entanglement cohomology from \cite{Mainiero:2019enr} which are relevant for our discussion here, including some small novel observations. In Section \ref{sec:hodge} we build a Hodge theory for entanglement cohomology, prove analogues of the Hodge theorem and Hodge decomposition, and construct a Hodge star operation which demonstrates that the Poincar\'e polynomials associated with the entanglement cohomology of generic pure states are symmetric. 
Section \ref{sec:two_qubit_example} demonstrates this machinery using an explicit comparison of the entanglement cohomologies and Hodge star operations for two example states involving pairs of qubits.
Section \ref{sec:conclusion} summarizes our results and presents directions for future research. Finally, we have relegated an argument concerning a certain compatibility condition for projections to Appendix \ref{app:compatibility}.

\section{Entanglement Cohomology}\label{sec:ent_coho_review}

In this section, we introduce the basic definitions and results on entanglement cohomology which will be useful in the remainder of this paper. Our discussion is primarily a review of \cite{Mainiero:2019enr}, where these ideas were first developed. However, our notation and focus will be a bit different than those of \cite{Mainiero:2019enr}, and we will include some minor new observations. We therefore find it useful to give a fairly substantial review of the machinery of entanglement cohomology, which also serves to make the present work more self-contained. In order to make this paper accessible to a broad audience, we will also review some elementary notions related to entanglement and density matrices; although these ideas are familiar to physicists, this discussion will serve to fix our notation, and may be useful for readers with a more mathematical background.

\subsection{Notation and Basic Definitions}\label{sec:defns}

Throughout this work, we restrict attention to Hilbert spaces which can be realized as finite tensor products of smaller, finite-dimensional Hilbert spaces, where we refer to the latter as subsystems. For instance, in the case of a bipartite Hilbert space, we write
\begin{align}
    \mathcal{H}_{AB} = \mathcal{H}_A \otimes \mathcal{H}_B \, ,
\end{align}
where $\mathcal{H}_A$ and $\mathcal{H}_B$ are subsystems with dimensions that we write as $d_A$ and $d_B$, respectively. If we choose two bases $\ket{a}$ and $\ket{b}$ for $\mathcal{H}_A$ and $\mathcal{H}_B$, a generic state $\ket{\psi_{AB}} \in \mathcal{H}_{AB}$ then admits an expansion
\begin{align}\label{bipartite_expansion}
    \ket{\psi_{AB}} = \sum_{a = 1}^{d_A} \sum_{b=1}^{d_B} C_{ab} \ket{a} \otimes \ket{b} \, ,
\end{align}
where $C_{ab}$ are a set of expansion coefficients.

The above notation is suitable for discussing Hilbert spaces with a small, fixed number of tensor product factors. However, when we consider Hilbert spaces with an arbitrary number of factors, it will be convenient to introduce some additional conventions in order to clarify the role of various indices. For example, if the index $i$ appears in a formula, it may not be immediately clear whether this refers to the $i$-th Hilbert space $\mathcal{H}_i$, or to the $i$-th basis element $\ket{ i }$ for some other Hilbert space $\mathcal{H}_j$ for $j \neq i$.

To avoid this confusion, we adopt the rule that indices which label subsystems will always be decorated with an underline, so that $\ul{i}$ refers to the $i$-th Hilbert space. When choosing indices for such subsystem labels, we will always use lowercase middle Latin letters like $\ul{i}$, $\ul{j}$, $\ul{k}$. A general Hilbert space of the type which we consider in this work is
\begin{align}
    \mathcal{H}_{\ul{1} \ldots \ul{n}} = \prod_{\ul{i} = \ul{1}}^{\ul{n}} \mathcal{H}_{\ul{i}} = \mathcal{H}_{\ul{1}} \otimes \mathcal{H}_{\ul{2}} \otimes \ldots \otimes \mathcal{H}_{\ul{n}} \, .
\end{align}
We choose a set of basis elements $\ket{ \alpha_\ul{i} }$ for each Hilbert space $\mathcal{H}_\ul{i}$, where the index $\alpha_{\ul{i}}$ runs from $1_{\ul{i}}$ up to $d_{\ul{i}}$, which denotes the dimension of the $\ul{i}$-th Hilbert space. For instance, if the third Hilbert space $\mathcal{H}_{\ul{3}}$ is $4$-dimensional, we would write its basis elements as
\begin{align}
    \ket{ 1_{\ul{3 } } } \, , \quad \ket{ 2_{\ul{3 } } } \, , \quad \ket{ 3_{\ul{3 } } } \, , \quad \ket{ 4_{\ul{3 } } } \, ,
\end{align}
which we will sometimes abbreviate using expressions like
\begin{align}\label{separate_increment}
    \left\{ \ket{\gamma_{\ul{3}}} \mid 1 \leq \gamma \leq 4 \right\} \, ,
\end{align}
where in expressions like (\ref{separate_increment}) it is understood that the integer variable $\gamma$ takes values that increment independently of the subscript $\ul{3}$ which labels the subsystem.  We always use Greek letters to index the individual basis elements within a given Hilbert space, in contrast with the underlined Latin letters that specify the Hilbert space.

The generalization of the expansion (\ref{bipartite_expansion}) to an arbitrary $n$-partite Hilbert space is
\begin{align}\label{uncondensed_general_ket}
    \ket{\psi_{\ul{1} \ldots \ul{n}}} = \sum_{\alpha_\ul{n} = 1_{\ul{n}}}^{d_\ul{n}} \ldots \sum_{\alpha_\ul{1} = 1_{\ul{1}}}^{d_\ul{1}} C_{\alpha_\ul{1} \cdots \alpha_\ul{n}} \ket{ \alpha_\ul{1}} \otimes \ldots \otimes \ket{ \alpha_\ul{n} } \, ,
\end{align}
although we will prefer to write such expansions using a condensed sum notation,
\begin{align}\label{condensed_general_ket}
    \ket{\psi} = \sum_{\alpha_\ul{i} = 1_{\ul{i}}}^{d_\ul{i}} C_{\alpha_\ul{1} \cdots \alpha_\ul{n}} \ket{ \alpha_\ul{1}} \otimes \ldots \otimes \ket{ \alpha_\ul{n} } \, ,
\end{align}
where the bounds of summation are understood to imply that we separately sum over each $\alpha_\ul{i}$ for $\ul{i} = \ul{1} , \ldots , \ul{n}$. In expressions like (\ref{uncondensed_general_ket}) and (\ref{condensed_general_ket}), we iterate over values of the \emph{combined} index $\alpha_{\ul{i}}$, rather than fixing $\ul{i}$ and iterating over values of $\alpha$, as in (\ref{separate_increment}). We trust that no confusion will occur from alternately using both of these summation conventions, as we will generally indicate the bounds of summation explicitly, and one can always distinguish which index refers to the subsystem due to the underlining.

We say that $\ket{\psi_{\ul{1} \ldots \ul{n}}} \in \mathcal{H}_{\ul{1} \ldots \ul{n}}$ is a \emph{product state} if it can be written as
\begin{align}
    \ket{\psi_{\ul{1} \ldots \ul{n}}} = \ket{\psi_\ul{1}} \otimes \ket{\psi_\ul{2}} \otimes \ldots \otimes \ket{\psi_\ul{n}} \, ,
\end{align}
for some collection of $\ket{\psi_\ul{i}} \in \mathcal{H}_\ul{i}$. If $\ket{\psi_{\ul{1} \ldots \ul{n}}} $ is not a product state, we say that it is \emph{entangled}.

From a state $\ket{\psi_{\ul{1} \ldots \ul{n}}} \in \Honeton$, one can construct the associated density matrix
\begin{align}\label{pure_state_density_matrix}
    \rho_{\ul{1} \ldots \ul{n}} = \ket{\psi_{\ul{1} \ldots \ul{n}}} \bra{\psi_{\ul{1} \ldots \ul{n}}} \, .
\end{align}
As we have done with states, we also decorate the symbols for density matrices with a collection of subscripts as in (\ref{pure_state_density_matrix}) that list the tensor product factors of the Hilbert space on which the density matrix acts. When we use alphabetical labels, in situations with a small and fixed number of tensor product factors, we will likewise include a collection of capital Latin subscripts on the density matrix, writing (for instance) $\rho_{AB}$ for a density matrix describing a state in $\mathcal{H}_{AB} = \mathcal{H}_A \otimes \mathcal{H}_B$. 

A density matrix of the form (\ref{pure_state_density_matrix}), which is a rank-one projector onto the state $\ket{\psi_{\ul{1} \ldots \ul{n}}}$, will be referred to as a \emph{pure state density matrix}. A generic density matrix $\rho$ can be written as a convex combination of pure state density matrices,
\begin{align}\label{general_density_matrix}
    \rho = \sum_{a = 1}^{N} p_a \ket{\psi_a} \bra{\psi_a} \, ,
\end{align}
where $p_a > 0$ for each $a$ and $\sum_a p_a = 1$. A density matrix which takes the form (\ref{general_density_matrix}) and involves $N > 1$ non-vanishing terms in the sum is referred to as a \emph{mixed state density matrix}. An equivalent definition is that a density matrix $\rho$ is a Hermitian, positive semi-definite operator acting on a Hilbert space $\mathcal{H}$ which satisfies $\Tr ( \rho ) = 1$; if in addition $\Tr ( \rho^2 ) = 1$ then $\rho$ is a pure state, or if $\Tr ( \rho^2 ) < 1$ then $\rho$ is a mixed state.

For any density matrix $\rho$, it will be convenient to define the \emph{support projection} operator $s_\rho$ which is a projector onto the image of $\rho$. In the case where $\rho$ is a pure state density matrix, it is clear that $\rho$ is already a projection operator, so one has $\rho = s_\rho$. More generally, given an arbitrary mixed state density matrix $\rho$, one may diagonalize to write
\begin{align}\label{diagonalized_rho}
    \rho = \sum_{a = 1}^{M} \lambda_a \ket{\chi_a} \bra{\chi_a} \, ,
\end{align}
where now the $\ket{\chi_a}$ are orthonormal, and the corresponding support projection $s_\rho$ is
\begin{align}\label{supp_proj_defn}
    s_\rho = \sum_{a=1}^{M} \ket{\chi_a} \bra{\chi_a} \, ,
\end{align}
which is obtained by simply replacing each coefficient in the sum (\ref{diagonalized_rho}) by unity.

An important construction for operators acting on multipartite Hilbert spaces is the \emph{partial trace}, which maps an operator acting on an $n$-partite Hilbert space to a ``reduced'' operator acting on an $(n-1)$-partite Hilbert space. Consider a linear operator
\begin{align}
    \mathcal{O}_{\ul{1} \ldots \ul{n}} : \mathcal{H}_{\ul{1} \ldots \ul{n}} \to \mathcal{H}_{\ul{1} \ldots \ul{n}} \, ,
\end{align}
which may (for example) be a density matrix. 

Any such operator admits an expansion, in the condensed sum notation of (\ref{condensed_general_ket}), as
\begin{align}\label{general_expansion}
    \mathcal{O}_{\ul{1} \ldots \ul{n}} = \sum_{\alpha_{\ul{i}} = 1_{\ul{i}}}^{d_{\ul{i}}} \sum_{\beta_{\ul{j}} = 1_{\ul{j}}}^{d_{\ul{j}}} \mathcal{O}_{\alpha_\ul{1} \cdots \alpha_\ul{n}}^{\beta_\ul{1} \cdots \beta_\ul{n}} \ket{\alpha_{\ul{1}}} \bra{\beta_{\ul{1}}} \otimes \ldots \otimes \ket{\alpha_{\ul{n}}} \bra{\beta_{\ul{n}}} \, .
\end{align}
Given a fixed subsystem $\underline{k}$, the partial trace of the operator $\mathcal{O}_{\ul{1} \ldots \ul{n}}$ is defined by
\begin{align}\label{ptrace_defn}
    \tr_{\ul{k}} \left( \mathcal{O}_{\ul{1} \ldots \ul{n}} \right) &= \sum_{\alpha_{\ul{i}} = 1_{\ul{i}}}^{d_{\ul{i}}} \sum_{\beta_{\ul{j}} = 1_{\ul{j}}}^{d_{\ul{j}}} \mathcal{O}_{\alpha_\ul{1} \cdots \alpha_\ul{n}}^{\beta_\ul{1} \cdots \beta_\ul{n}} \ket{\alpha_{\ul{1}}} \bra{\beta_{\ul{1}}} \otimes \ldots \tr \Big( | \alpha_{\ul{k}} \rangle \, \bra{\beta_{\ul{k}}} \Big) \ldots \otimes \ket{\alpha_{\ul{n}}} \bra{\beta_{\ul{n}}} \nonumber \\
    &= \sum_{\alpha_{\ul{i}} = 1_{\ul{i}}}^{d_{\ul{i}}} \sum_{\beta_{\ul{j}} = 1_{\ul{j}}}^{d_{\ul{j}}} \mathcal{O}_{\alpha_\ul{1} \cdots \alpha_\ul{n}}^{\beta_\ul{1} \cdots \beta_\ul{n}} \ket{\alpha_{\ul{1}}} \bra{\beta_{\ul{1}}} \otimes \ldots \Big( \sum_{\gamma_{\ul{k}} = 1_{\ul{k}}}^{d_{\ul{k}}}  \braket{\gamma_{\ul{k}}}{\alpha_{\ul{k}}} \braket{\beta_{\ul{k}}}{\gamma_{\ul{k}}} \Big)  \ldots \otimes \ket{\alpha_{\ul{n}}} \bra{\beta_{\ul{n}}} \, .
\end{align}
On the right side of the first line of (\ref{ptrace_defn}), the symbol $\tr$ denotes the ordinary trace over the Hilbert space $\mathcal{H}_{\ul{k}}$, which we have written in terms of a sum over basis elements in the second line. The result of this operation is a new operator $\mathcal{O}_{\ul{1} \ldots \ul{k-1} \, \ul{k+1} \ldots \ul{n}}$ which acts on the $(n-1)$ partite Hilbert space $\mathcal{H}_{\ul{1} \ldots \ul{k-1} \, \ul{k+1} \ldots \ul{n}}$ that no longer includes the subsystem $\ul{k}$. We use the symbol $\slashed{\ul{k}}$ to indicate that a particular index has been omitted from a list, e.g.
\begin{align}\label{no_k_defn}
    \mathcal{H}_{\ul{1} \ldots \slashed{\ul{k}} \ldots \ul{n}} = \mathcal{H}_{\ul{1}} \otimes \ldots \otimes \mathcal{H}_{\ul{k-1}} \otimes \mathcal{H}_{\ul{k+1}} \ldots \otimes \mathcal{H}_{\ul{n}} \, ,
\end{align}
where the right side of (\ref{no_k_defn}) is meant to indicate that the subsystem $\ul{k}$ does not appear in the tensor product. Likewise, we will write $\slashed{\alpha}_{\ul{k}}$ if the basis vectors associated with the subsystem $\ul{k}$ are excluded from a list or sum.

Another way to characterize the partial trace $\tr_{\ul{k}} \left( \mathcal{O}_{\ul{1} \ldots \ul{n}} \right)$ is to give its expansion coefficients (\ref{general_expansion}) in a basis for the reduced Hilbert space. That is, if one writes
\begin{align}
    \tr_{\ul{k}} \left( \mathcal{O}_{\ul{1} \ldots \ul{n}} \right) = \sum_{\alpha_{\ul{i}} = 1_{\ul{i}}}^{d_{\ul{i}}} \sum_{\beta_{\ul{j}} = 1_{\ul{j}}}^{d_{\ul{j}}} \mathcal{O}_{\alpha_\ul{1} \cdots \slashed{\alpha}_{\ul{k}} \cdots \alpha_\ul{n}}^{\beta_\ul{1} \cdots \slashed{\beta}_{\ul{k}} \cdots \beta_\ul{n}} \ket{\alpha_{\ul{1}}} \bra{\beta_{\ul{1}}} \otimes \ldots \otimes | \slashed{\alpha}_{\ul{k}} \rangle \, \langle \slashed{\beta}_{\ul{k}} | \otimes \ldots \otimes \ket{\alpha_{\ul{n}}} \bra{\beta_{\ul{n}}} \, ,
\end{align}
then the coefficients of the reduced operator are simply
\begin{align}
    \mathcal{O}_{\alpha_\ul{1} \cdots \slashed{\alpha}_{\ul{k}} \cdots \alpha_\ul{n}}^{\beta_\ul{1} \cdots \slashed{\beta}_{\ul{k}} \cdots \beta_\ul{n}} = \sum_{\alpha_{\ul{k}} = 1_{\ul{k}}}^{d_{\ul{k}}} \mathcal{O}_{\alpha_\ul{1} \cdots \alpha_{\ul{k}} \cdots \alpha_\ul{n}}^{\beta_\ul{1} \cdots \alpha_{\ul{k}} \cdots \beta_\ul{n}} \, .
\end{align}
Admittedly, our notation for the partial trace is somewhat cumbersome, but this is primarily because we have formulated our definitions in a way which applies to general tensor product Hilbert spaces with any number of factors possessing arbitrary dimensions. As an example, let us briefly see how these formulas simplify in the case of a bipartite Hilbert space $\mathcal{H}_{AB} = \mathcal{H}_A \otimes \mathcal{H}_B$. In this case, a general operator $\mathcal{O}_{AB}$ can be written as a linear combination of basis operators of the form
\begin{align}\label{basis_operator}
    \mathfrak{o}_{AB} = \ket{a_1} \bra{a_2} \otimes \ket{b_1} \bra{b_2} \, ,
\end{align}
where $\ket{a_1}$, $\ket{a_2} \in \mathcal{H}_A$ and $\ket{b_1}$, $\ket{b_2} \in \mathcal{H}_B$. One can define the partial trace by giving its value for operators of the form (\ref{basis_operator}),
\begin{align}
    \tr_B \left( \mathfrak{o}_{AB}  \right) = \ket{a_1} \bra{a_2} \tr \left( \ket{b_1} \bra{b_2} \right) \, ,
\end{align}
and then extending this definition to general operators by linearity. This construction agrees with the general definition (\ref{ptrace_defn}) in this case.

Clearly this procedure of partial tracing can be iterated to sequentially trace out multiple subsystems. Generically, given an operator defined on some multi-partite Hilbert space -- such as a density matrix $\rho_{\ul{1} \ldots \ul{n}}$ which acts on $\mathcal{H}_{\ul{1} \ldots \ul{n}}$ -- we will indicate partial traces of this operator using the same symbol but omitting the indices that correspond to subsystems that have been traced out. For instance,
\begin{align}
    \rho_{\ul{1} \ul{2}} = \tr_{\ul{3}} \left( \rho_{\ul{1} \ul{2} \ul{3}} \right) \, , \qquad \rho_{\ul{1}} = \tr_{\ul{2}} \left(  \rho_{\ul{1} \ul{2}} \right) \, , \qquad \text{etc.}
\end{align}
In our subsequent analysis, it will be important to understand the interplay between a state in a Hilbert space $\mathcal{H}$, which is described by a density matrix $\rho$, and the collection of linear operators $\mathcal{O} : \mathcal{H} \to \mathcal{H}$ acting on that Hilbert space. A useful notion in describing this interplay is the ``restriction'' of the action of such an operator $\mathcal{O}$ to the image of the density matrix. We will denote such a restriction with a vertical bar,
\begin{align}\label{restriction}
    \mathcal{O} \big\vert_{\rho} = s_\rho \mathcal{O} s_\rho \, ,
\end{align}
where $s_\rho$ is defined in (\ref{supp_proj_defn}). The object $\mathcal{O} \big\vert_{\rho}$ is a new linear operator acting on $\mathcal{H}$ which first projects an input state onto the image $\mathrm{im} ( \rho )$, then acts with the operator $\mathcal{O}$, and again projects this output onto the image $\mathrm{im} ( \rho )$. This composition of maps therefore describes the action of this operator $\mathcal{O}$ on $\im ( \rho )$. More explicitly, the restriction $\mathcal{O} \big\vert_{\rho}$ describes the action of $\mathcal{O}$ on the subspace of $\mathcal{H}$ spanned by the kets $\ket{\psi_a}$ which appear in the admixture (\ref{general_density_matrix}) defining the density matrix.

Given a multi-partite Hilbert space $\Honeton$ and a density matrix $\rho_{\ul{1} \ldots \ul{n}}$, it will be convenient to define the space of all such restricted operators $\mathcal{O}_{\ul{1} \ldots \ul{n}} \big\vert_{\rho_{\ul{1} \ldots \ul{n}}}$ associated with this density matrix, which we will denote as
\begin{align}\label{better_commutant_definition}
    \Omega^n \left( \rho_{\ul{1} \ldots \ul{n}} \right) = \left\{ \mathcal{O}_{\ul{1} \ldots \ul{n}} \big\vert_{\rho_{\ul{1} \ldots \ul{n}}} \mid \mathcal{O}_{\ul{1} \ldots \ul{n}} \text{ is a linear operator on } \Honeton \right\} \, .
\end{align}
We have chosen the notation $\Omega^n$ in (\ref{better_commutant_definition}) to remind the reader of the space $\Omega^n ( \mathcal{M} )$ of all differential $n$-forms defined on a manifold $\mathcal{M}$. Note that the symbol $\Omega$ is labeled with an $n$ to emphasize that the relevant Hilbert space is a tensor product of $n$ subsystems.

Given the definition (\ref{better_commutant_definition}), one might wonder whether there is a natural notion of $\Omega^k \left( \rho_{\ul{1} \ldots \ul{n}} \right)$ for $1 \leq k < n$. We will define these spaces of operators as follows. For each $k$, first consider the collection of all ${n \choose k}$ ways of performing sequential partial trace operations on $\rho_{\ul{1} \ldots \ul{n}}$ to obtain a reduced density matrix $\rho_{\ul{i}_1 \ldots \ul{i}_k}$ acting on a $k$-partite Hilbert space. Given each such reduced system, one can consider the space of restricted operators acting on $\mathcal{H}_{\ul{i}_1 \ldots \ul{i}_k}$, in exactly the same way as in equation (\ref{better_commutant_definition}):
\begin{align}\label{reduced_commutant}
    \Omega^k \left( \rho_{\ul{i}_1 \ldots \ul{i}_k} \right) = \left\{ \mathcal{O}_{\ul{i}_1 \ldots \ul{i}_k} \big\vert_{\rho_{\ul{i}_1 \ldots \ul{i}_k}} \mid \mathcal{O}_{\ul{i}_1 \ldots \ul{i}_k} \text{ is a linear operator on } \mathcal{H}_{\ul{i}_1 \ldots \ul{i}_k} \right\} \, .
\end{align}
We then define $\Omega^k \left( \rho_{\ul{1} \ldots \ul{n}} \right)$ as the collection of all tuples of elements of the ${n \choose k}$ spaces (\ref{reduced_commutant}), as the Hilbert spaces $\mathcal{H}_{\ul{i}_1 \ldots \ul{i}_k}$ run over all length-$k$ subsets of the $n$ subsystems. That is,
\begin{align}\label{lower_commutant_definition}
    \Omega^k \left( \rho_{\ul{1} \ldots \ul{n}} \right) = \bigtimes_{ \{ \ul{i}_1 , \ldots , \ul{i}_k \} \, \subset \, \{ \ul{1} , \ldots , \ul{n} \} } \Omega^k \left( \rho_{\ul{i}_1 \ldots \ul{i}_k} \right) \, .
\end{align}
In practice, we will arrange the tuple of elements in any space (\ref{lower_commutant_definition}) using a lexicographic ordering of the subsystems. We will sometimes use underlined capital Latin letters to represent multi-indices, such as $\ul{I} = ( \ul{i}_1 , \ldots , \ul{i}_k )$, where we assume that the entries of such a multi-index are in increasing order, $\ul{i}_1 < \ldots < \ul{i}_n$. Given a multi-index $\ul{I}$, we write $\ul{I}^C$ for its complement in the full set of subsystems $\ul{1} , \ldots , \ul{n}$, again typically in increasing order.

An example might serve to make these definitions clearer. Suppose that we begin with a tripartite Hilbert space $\mathcal{H}_{ABC} = \mathcal{H}_A \otimes \mathcal{H}_B \otimes \mathcal{H}_C$, a state $\ket{\psi_{ABC}} \in \mathcal{H}_{ABC}$, and the associated density matrix $\rho_{ABC} = \ket{\psi_{ABC}} \bra{\psi_{ABC}}$. We use the simplified notation $s_{ABC}$ for the support projection operator $s_{\rho_{ABC}}$ that projects onto the image of $\rho_{ABC}$. One can then assemble the various reduced density matrices
\begin{gather}
    \rho_{AB} = \tr_C \left( \rho_{ABC} \right) \, , \quad \rho_{AC} = \tr_B \left( \rho_{ABC} \right) \, , \quad \rho_{BC} = \tr_A \left( \rho_{ABC} \right) \, , \nonumber \\
    \rho_A = \tr_B \left( \rho_{AB} \right) \, , \quad \rho_B = \tr_A \left( \rho_{AB} \right) \, , \quad \rho_C = \tr_B \left( \rho_{BC} \right) \, , \label{tripartite_reductions}
\end{gather}
along with all of the corresponding support projection operators $s_{AB}$, $s_{AC}$, and so on.

An element of $\Omega^3 ( \rho_{ABC} )$ is the restriction of a linear operator $\mathcal{O}_{ABC}$ acting on $\mathcal{H}_{ABC}$,
\begin{align}
    \omega_3 \in \Omega^3 ( \rho_{ABC} ) \; \implies \; \omega_3 = \mathcal{O}_{ABC} \big\vert_{\rho_{ABC}} = s_{ABC} \mathcal{O}_{ABC} s_{ABC} \, .
\end{align}
Next, a two-form $\omega_2 \in \Omega^2 ( \rho_{ABC} )$ is a tuple of three restricted operators acting on the three subsystems $\mathcal{H}_{AB}$, $\mathcal{H}_{AC}$, and $\mathcal{H}_{BC}$ of size $2$:
\begin{align}
    \omega_2 &= \left( \mathcal{O}_{AB} \big\vert_{\rho_{AB}} \, , \, \mathcal{O}_{AC} \big\vert_{\rho_{AC}} \, , \, \mathcal{O}_{BC} \big\vert_{\rho_{BC}} \right) \nonumber \\
    &= \left( s_{AB} \mathcal{O}_{AB} s_{AB} \, , \, s_{AC} \mathcal{O}_{AC} s_{AC} \, , \, s_{BC} \mathcal{O}_{BC} s_{BC} \right) \, .
\end{align}
Similarly, a one-form $\omega_1 \in \Omega^1 ( \rho_{ABC} )$ is a tuple of three operators acting on the size-one subsystems $\mathcal{H}_A$, $\mathcal{H}_B$, and $\mathcal{H}_C$:
\begin{align}
    \omega_1 &= \left( \mathcal{O}_{A} \big\vert_{\rho_{A}} \, , \, \mathcal{O}_{B} \big\vert_{\rho_{B}} \, , \, \mathcal{O}_{C} \big\vert_{\rho_{C}} \right) \nonumber \\
    &= \left( s_{A} \mathcal{O}_{A} s_{A} \, , \, s_{B} \mathcal{O}_{B} s_{B} \, , \, s_{C} \mathcal{O}_{C} s_{C} \right) \, .
\end{align}
Therefore, we see that the objects of the spaces $\Omega^k \left( \rho_{\ul{i}_1 \ldots \ul{i}_k} \right)$, which we will refer to as entanglement $k$-forms, are quite simple objects. They are simply length-$k$ lists of finite-dimensional matrices, each of the appropriate dimension to act on some reduced Hilbert space, and which have been sandwiched with the suitable support projection operator $s_\rho$.

Finally, let us remind the reader of some elementary definitions from homological algebra. We will be interested in complexes which are formed from a sequence of vector spaces\footnote{Complexes, homology, and cohomology can, of course, be defined for more general algebraic structures such as modules. However, for our purposes in this work, it is sufficient to restrict to vector spaces.} connected by maps which ``square to zero'' in the same way that the exterior derivative $d$ is nilpotent when acting on differential forms. Abstractly, we say that a \emph{cochain complex} is a collection of vector spaces $V_n$ assembled into a structure
\begin{align}\label{cochain_defn}
    \cdots \xrightarrow{d^{-1}} V_0 \xrightarrow{d^{0}} V_1 \xrightarrow{d^{1}} V_2 \xrightarrow{d^{2}} \ldots \, ,
\end{align}
with the property that $d^{n+1} \circ d^n = 0$.

Given such a structure, one defines the $n$-th cohomology group as
\begin{align}\label{cohomology_defn}
    H^n = \frac{\ker ( d^n ) }{\mathrm{im} ( d^{n-1} ) } \, ,
\end{align}
where $\ker$ and $\mathrm{im}$ denote the kernel and image, respectively. 

The dual of this construction is called a \emph{chain complex}, which likewise involves a collection of vector spaces, but which are connected by operators $d_n$ written with lower indices and which act in the opposite direction:
\begin{align}\label{chain_defn}
    \cdots \xleftarrow{d_0} V_0 \xleftarrow{d_{1}} V_1 \xleftarrow{d_2} V_2 \xleftarrow{d_3} \ldots \, ,
\end{align}
which now satisfy $d_n \circ d_{n+1} = 0$. In this case, one speaks of the $n$-th homology group,
\begin{align}
    H_n = \frac{\ker ( d_n ) }{\mathrm{im} ( d_{n+1} ) } \, .
\end{align}

\subsection{Entanglement Complex}

We are now in a position to define the main object of study in this work, which is a certain cochain complex associated to any density matrix in a multi-partite Hilbert space.

In order to keep the presentation pedagogical, we will work up to the general definition of this object in steps. We first describe the entanglement complex in the simplest case of a bipartite Hilbert space $\mathcal{H}_{AB} = \mathcal{H}_A \otimes \mathcal{H}_B$. Next we will present the corresponding definition for a tripartite Hilbert space $\mathcal{H}_{ABC}$, which naturally leads us to address a new subtlety related to ordering. Finally, we then present the general definition of the entanglement complex for an arbitrary multi-partite Hilbert space $\Honeton$.

\subsubsection*{\ul{\it Bipartite complexes}}

Consider a bipartite Hilbert space $\mathcal{H}_{AB}$, a state $\ket{\psi_{AB}}$ with associated density matrix $\rho_{AB} = \ket{\psi_{AB}} \bra{\psi_{AB}}$ and support projection $s_{AB} = s_{\rho_{AB}}$, and the partial traces
\begin{align}
    \rho_A = \tr_B \left( \rho_{AB} \right) \, , \qquad \rho_B = \tr_A \left( \rho_{AB} \right) \, ,
\end{align}
where again we write $s_A = s_{\rho_A}$ and $s_B = s_{\rho_B}$ for the corresponding support projections.

Following the definitions of Section \ref{sec:defns}, we can consider the space $\Omega^2 ( \rho_{AB} )$ whose elements are linear operators $\mathcal{O}_{AB} \big\vert_{\rho_{AB}}$ acting on $\mathcal{H}_{AB}$ that have been left- and right-multiplied by $s_{AB}$. Similarly, the space $\Omega^1 ( \rho_{AB} )$ consists of all ordered pairs $(\mathcal{O}_A, \mathcal{O}_B)$ where $\mathcal{O}_A = \mathcal{O}_A \big\vert_{\rho_A}$ is a linear operator on $\mathcal{H}_A$ restricted to the image of $\rho_A$, and likewise $\mathcal{O}_B = \mathcal{O}_B \big\vert_{\rho_B}$ is an analogous object associated with $\mathcal{H}_B$.

We now propose to assemble these spaces into a cochain complex
\begin{align}\label{bipartite_complex_defn}
    0 \to \mathbb{C} \xrightarrow{d^{0}} \Omega^1 ( \rho_{AB} ) \xrightarrow{d^{1}} \Omega^2 ( \rho_{AB} ) \xrightarrow{d^2} 0 \, ,
\end{align}
which is a structure of the form (\ref{cochain_defn}). Throughout this article, we label the coboundary operators with degrees that are one larger than the corresponding labels in \cite{Mainiero:2019enr}; for instance, what we call $d^0$ in (\ref{bipartite_complex_defn}) is $d^{-1}$ in \cite{Mainiero:2019enr}, and our $d^1$ corresponds to $d^0$ of that work. We will sometimes suppress the indices on the operators, writing the same symbol $d$ for both $d^{0}$ and $d^1$, when it is clear from context which is intended.

To specify the complex (\ref{bipartite_complex_defn}), we must define the action of the coboundary operators $d^0$ and $d^1$ in a way which obeys $d^1 \circ d^0 = 0$, and which (ideally) captures information about entanglement in the state $\ket{\psi_{AB}}$ which defines $\rho_{AB}$. Definitions of the coboundary maps satisfying these two properties were written down in \cite{Mainiero:2019enr}; first one declares
\begin{align}\label{d0_bipartite}
    d^0 \lambda = ( \lambda s_A, \lambda s_B ) \, ,
\end{align}
for any $\lambda \in \mathbb{C}$. Note that $( \lambda s_A, \lambda s_B )$ is a valid element of $\Omega^1 ( \rho_{AB} )$, because both projection operators are automatically restricted to the images of the corresponding density matrices. Said differently, these operators are invariant under the bar operation $\big\vert_{\rho}$, since
\begin{align}
    s_A \big\vert_{\rho_A} = s_A s_A s_A = s_A \, ,
\end{align}
and likewise for $s_B$, since any projection operator $P$ satisfies $P^2 = P$ by definition.

Next, we define the action of $d^1$ on a tuple $(\mathcal{O}_A, \mathcal{O}_B) \in \Omega^1$ by
\begin{align}\label{d1_defn}
    d^1 ( \mathcal{O}_A , \mathcal{O}_B ) = \left( \mathbb{I}_A \otimes \mathcal{O}_B - \mathcal{O}_A \otimes \mathbb{I}_B \right) \big\vert_{\rho_{AB}} \, .
\end{align}
Here we write $\mathbb{I}_A$ for the identity operator acting on $\mathcal{H}_A$ and likewise for $\mathbb{I}_B$. The definition (\ref{d1_defn}) is reminiscent of the formula for the exterior derivative of a one-form $\omega = \omega_x \, dx + \omega_y \, dy$ on a $2d$ manifold with coordinates $(x, y)$, namely
\begin{align}
    d \omega = \left( \partial_x \omega_y - \partial_y \omega_x \right) \, dx \wedge dy \, ,
\end{align}
where the roles of the component functions $\omega_x$, $\omega_y$ are played by the operators $\mathcal{O}_A$ and $\mathcal{O}_B$, and the action of the partial derivatives $\partial_x$, $\partial_y$ is replaced by the action of taking the tensor product with the identity operators $\mathbb{I}_A$ and $\mathbb{I}_B$ in the subsystem Hilbert spaces.

First we should verify that this definition obeys $d d = 0$, where we have suppressed indices.\footnote{We choose to abbreviate $d^{n+1} \circ d^n$ as $d d$ rather than $d^2$, to avoid confusion about whether the superscript ${}^{2}$ indicates squaring or whether it labels the second coboundary operator $d^2$.} Composing the two coboundary maps defined above, one finds
\begin{align}
    d d \lambda = \lambda s_{AB} \left( \mathbb{I}_A \otimes s_B - s_A \otimes \mathbb{I}_B \right) s_{AB} \, , \qquad \lambda \in \mathbb{C} \, .
\end{align}
However, the support projection operators associated with any density matrix $\rho_{AB}$ and its partial traces $\rho_A = \tr_B \left( \rho_{AB} \right)$, $\rho_B = \tr_A \left( \rho_{AB} \right)$ necessarily obey
\begin{align}\label{compatibility_bipartite}
   \left( s_A \otimes s_B \right) s_{AB} = s_{AB} =  s_{AB} \left( s_A \otimes s_B \right) \, ,
\end{align}
which is a condition that is called ``compatibility of supports'' in \cite{Mainiero:2019enr}, where this result is proven using an abstract algebraic argument. We have provided a more pedestrian argument for this compatibility condition in Appendix \ref{app:compatibility} assuming that the parent state on $\mathcal{H}_{AB}$ is pure (although the result holds more generally). Using (\ref{compatibility_bipartite}), one finds
\begin{align}\label{dd_equals_zero_bipartite}
    d d \lambda &= \lambda s_{AB} \left( \mathbb{I}_A \otimes s_B - s_A \otimes \mathbb{I}_B \right) s_{AB} \nonumber \\
    &= \lambda s_{AB} \left( s_A \otimes s_B \right) \left( \mathbb{I}_A \otimes s_B - s_A \otimes \mathbb{I}_B \right) \left( s_A \otimes s_B \right) s_{AB} \nonumber \\
    &= \lambda s_{AB} \left( s_A \otimes s_B^2 - s_A^2 \otimes s_B \right) s_{AB} \nonumber \\
    &= 0 \, ,
\end{align}
where we have used the properties $s_A^2 = s_A$ and $s_B^2 = s_B$ of projectors.

This simple check confirms that the structure (\ref{bipartite_complex_defn}) satisfies the mathematical definition of a cochain complex, and thus one can consistently speak of its cohomology groups (\ref{cohomology_defn}). Specifically, we will speak only of the first cohomology group $H^1 ( \rho_{AB} )$, since we have defined only two non-trivial coboundary operators $d^1$ and $d^0$.

Although the entanglement complex, as we shall call (\ref{bipartite_complex_defn}), is mathematically well-defined, it is not yet clear that it has anything to do with entanglement. In Section \ref{sec:two_qubit_example}, we will consider an extended example comparing the cohomology groups $H^1 ( \rho_{AB} )$ for two examples of density matrices, one corresponding to a product state $\ket{\psi^{(P)}}$ and one corresponding to an EPR pair $\ket{\psi^{(E)}}$, and we will find that the cohomology is trivial for $\ket{\psi^{(P)}}$ and non-trivial for $\ket{\psi^{(E)}}$. This is a special case of a more general theorem which was proven in \cite{Mainiero:2019enr}. To state this theorem, let us first recall that any state $\ket{\psi_{AB}}$ which belongs to a bipartite Hilbert space $\mathcal{H}_{AB} = \mathcal{H}_A \otimes \mathcal{H}_B$ admits a \emph{Schmidt decomposition}
\begin{align}\label{schmidt}
    \ket{\psi_{AB}} = \sum_{\gamma=1}^{\min ( d_A, d_B )} \lambda_\gamma \ket{ \gamma_A } \otimes \ket{\gamma_B} \, ,
\end{align}
for some orthonormal bases $\left\{ \ket{\gamma_A} \right\}$ for $\mathcal{H}_A$ and $\left\{ \ket{\gamma_B} \right\}$ for $\mathcal{H}_B$, and some non-negative real coefficients $\lambda_\gamma$ which obey $\sum_\gamma \lambda_\gamma^2 = 1$. The number $S$ of non-zero coefficients $\lambda_\gamma$ in the sum (\ref{schmidt}) is called the \emph{Schmidt rank} of the state $\ket{\psi_{AB}}$.\footnote{To be clear, given \emph{any} orthonormal bases for $\mathcal{H}_A$ and $\mathcal{H}_B$, the state $\ket{\psi_{AB}}$ trivially admits an expansion of the form (\ref{bipartite_expansion}) in terms of $d_A \cdot d_B$ expansion coefficients $C_{ab}$. In contrast, the content of (\ref{schmidt}) is that there exist \emph{special} orthonormal bases for $\mathcal{H}_A$ and $\mathcal{H}_B$ in terms of which one can achieve a decomposition which involves only a smaller number, $\min ( d_A, d_B )$, of expansion coefficients $\lambda_\gamma$.}

In terms of the Schmidt rank $S$ of the state $\ket{\psi_{AB}}$ determining $\rho_{AB}$, one can show that
\begin{align}
    \dim \left( H^1 \left( \rho_{AB} \right) \right) = 2 \left( S^2 - 1 \right) \, .
\end{align}
This means that the cohomology is trivial if $S = 1$, which corresponds to a product state $\ket{\psi_{AB}} = \lambda_1 \ket{1_A} \otimes \ket{1_B}$, while the cohomology is non-trivial if $S > 1$ non-trivial Schmidt coefficients are required to express the state $\ket{\psi_{AB}}$, which is precisely the statement that $\ket{\psi_{AB}}$ is entangled. Thus the dimension of the entanglement cohomology indeed detects entanglement, at least in the case of a bipartite finite-dimensional Hilbert space $\mathcal{H}_{AB}$.

\subsubsection*{\ul{\it Tripartite complexes}}

Suppose that we are interested in the analogous cohomology for a tripartite Hilbert space,
\begin{align}
    \mathcal{H}_{ABC} = \mathcal{H}_A \otimes \mathcal{H}_B \otimes \mathcal{H}_C \, .
\end{align}
Fix a pure state density matrix $\rho_{ABC}$, along with its support projection $s_{ABC}$ and all of the corresponding objects for the various reduced Hilbert spaces, following the notation in (\ref{tripartite_reductions}) and the paragraph before it. In this setting, we would like to define a complex
\begin{align}\label{tripartite_complex_defn}
    0 \xrightarrow{\; \;} \mathbb{C} \xrightarrow{d^0} \Omega^1 \left( \rho_{ABC} \right) \xrightarrow{d^1} \Omega^2 \left( \rho_{ABC} \right) \xrightarrow{d^2} \Omega^3 \left( \rho_{ABC} \right) \xrightarrow{d^3} 0 \, .
\end{align}
Let us describe each of the coboundary operatators in turn. The first one, $d^0$, follows the recipe which we have seen in equation (\ref{d0_bipartite}) for the bipartite case:
\begin{align}
    d^0 \lambda = ( \lambda s_A , \lambda s_B , \lambda s_C ) \, .
\end{align}
Likewise, given a tuple $( \mathcal{O}_A , \mathcal{O}_B , \mathcal{O}_C )$ of three operators acting on the single-party Hilbert spaces $\mathcal{H}_A$, $\mathcal{H}_B$, and $\mathcal{H}_C$, respectively -- which have been appropriately restricted, in the sense that $\mathcal{O}_A = \mathcal{O}_A \big\vert_{\rho_A}$, $\mathcal{O}_B = \mathcal{O}_B \big\vert_{\rho_B}$, and $\mathcal{O}_C = \mathcal{O}_C \big\vert_{\rho_C}$ -- we define the operator $d^1$ by
\begin{align}\label{tripartite_d1}
    &d^1 ( \mathcal{O}_A , \mathcal{O}_B, \mathcal{O}_c ) \nonumber \\
    &\quad = \left( \left( \mathbb{I}_A \otimes \mathcal{O}_B - \mathcal{O}_A \otimes \mathbb{I}_B \right) \big\vert_{\rho_{AB}} \, , \, \left( \mathbb{I}_A \otimes \mathcal{O}_C - \mathcal{O}_A \otimes \mathbb{I}_C \right) \big\vert_{\rho_{AC}} \, , \, \left( \mathbb{I}_B \otimes \mathcal{O}_C - \mathcal{O}_B \otimes \mathbb{I}_C \right) \big\vert_{\rho_{BC}} \right) \, .
\end{align}
The right side of (\ref{tripartite_d1}) is a tuple of three operators acting on $\mathcal{H}_{AB}$, $\mathcal{H}_{AC}$, and $\mathcal{H}_{BC}$, all suitably restricted to the images of the corresponding reduced density matrices, which therefore defines a valid element of $\Omega^2 ( \rho_{ABC} )$. In particular, let us note how the signs have been chosen in (\ref{tripartite_d1}). Each term appears with a plus sign if the subsystem acted upon by the identity operator in that term appears \emph{before} the subsystem associated with the non-identity operator, in lexicographic ordering. For instance, $\mathbb{I}_B \otimes \mathcal{O}_C$ appears with a positive sign since $B$ precedes $C$ in the ordering $(A, B, C)$. Terms in which this order is reversed enter with a minus sign, such as $\mathcal{O}_B \, \otimes \, \mathbb{I}_C$, which contributes with a negative sign since $C$ does not precede $B$. This is the same pattern as in the bipartite case.

Finally, let us consider $d^2$, which will motivate us to introduce some additional notation. Let us begin with a tuple $\omega_2 \in \Omega^2 \left( \rho_{ABC} \right)$ with components
\begin{align}
    \omega_2 = \left( \mathcal{O}_{AB} , \mathcal{O}_{AC} , \mathcal{O}_{BC} \right) \, ,
\end{align}
where we again implicitly assume $\mathcal{O}_{AB} = \mathcal{O}_{AB} \big\vert_{\rho_{AB}}$ and so on. A natural guess for the action of $d^2$ might include three terms of the schematic form
\begin{align}\label{problematic}
    d^2 \omega_2 \overset{?}{=} \left( \mathbb{I}_A \otimes \mathcal{O}_{BC} + \mathcal{O}_{AC} \otimes \mathbb{I}_B + \mathcal{O}_{AB} \otimes \mathbb{I}_C \right) \big\vert_{\rho_{ABC}} \, ,
\end{align}
up to a choice of sign for each term which we postpone for the moment.

However, the expression (\ref{problematic}) is problematic for a couple of reasons. One issue is the middle term: the object $\mathcal{O}_{AC} \otimes \mathbb{I}_B$ is a linear operator acting on the Hilbert space $\mathcal{H}_A \otimes \mathcal{H}_C \otimes \mathcal{H}_B$, whereas we require an operator acting on these tensor product factors in the different order $\mathcal{H}_A \otimes \mathcal{H}_B \otimes \mathcal{H}_C$. A second problem is that, when we later define a notion of wedge product for entanglement forms, the definition (\ref{problematic}) will not satisfy a conventional Leibniz rule. The reason is that -- in the setting of differential forms -- taking a derivative of a product of components of one-forms, like $\partial_x \left( \omega_y \omega_z \right)$, generates two terms by the product rule. In contrast, taking a tensor product like $\mathbb{I}_A \otimes \left( \mathcal{O}_B \otimes \mathcal{O}_C \right)$, only one term is generated. We will remedy both of these problems in what follows.

We begin with the ordering issue. By analogy with de Rham cohomology for differential forms, whose similarity to entanglement cohomology we have been emphasizing, one might suspect that we should construct an operation similar to the wedge product $\wedge$. Just as the wedge allows us to rearrange products like $dx \wedge dz \wedge dy = - dx \wedge dy \wedge dz$ into a form with a canonical ordering $(x, y, z)$, while keeping track of minus signs that arise due to the signature of permutations, the desired wedge-like tensor product operation should allow us to rearrange expressions like the problematic term $\mathcal{O}_{AC} \otimes \mathbb{I}_B$ appearing in (\ref{problematic}).

Therefore, let us define an operation $\wotimes$ (where the decoration \, $\widehat{ {} }$ \, is meant to remind the reader of the wedge product $\wedge$) which acts by ``shuffling'' a tensor product of operators so that they act on the Hilbert space $\mathcal{H}_{ABC}$ with the tensor product factors in the correct order, while inserting an overall sign corresponding to the signature of the permutation required to implement the ordering.

We can define the action of $\wotimes$ explicitly, on a general multi-partite Hilbert space, as follows. First assume that we have fixed an ordering of the tensor product factors $\ul{1}$, $\ldots$, $\ul{n}$ of $\Honeton$. Consider two operators $\mathcal{A}_{\ul{i}_1 \ldots \ul{i}_a}$ and $\mathcal{B}_{\ul{j}_1 \ldots \ul{j}_b}$ which take tensor product forms,
\begin{align}\label{O_and_tilde_O_factorized}
    \mathcal{A}_{\ul{i}_1 \ldots \ul{i}_a} &= \mathcal{A}_{\ul{i}_1} \otimes \ldots \otimes \mathcal{A}_{\ul{i}_a} \, , \nonumber \\
    \mathcal{B}_{\ul{j}_1 \ldots \ul{j}_b} &= \mathcal{B}_{\ul{j}_1} \otimes \ldots \otimes \mathcal{B}_{\ul{j}_b} \, ,
\end{align}
and which act on Hilbert spaces which include the factors $\ul{i}_1, \ldots,  \ul{i}_a$ and $\ul{j}_1$, $\ldots$, $\ul{j}_b$, respectively. We assemble the collection of all subsystems into a list which is in increasing order; that is, we define a set of subsystem indices
\begin{align}\label{k_indices}
    \ul{k}_1 < \ul{k}_2 <  \ldots < \ul{k}_{a + b} \, ,
\end{align}
such that the list of indices $\ul{k}_1$, $\ldots$, $\ul{k}_{a + b}$ is a permutation $\sigma$ of the list $\ul{i}_1$, $\ldots$,  $\ul{i}_a$, $\ul{j}_1$, $\ldots$, $\ul{j}_b$. We use the symbol $\mathcal{C}_{\ul{x}}$ to designate an operator with is either one of the $\mathcal{A}_{\ul{x}}$ or one of the $\mathcal{B}_{\ul{x}}$, depending on the value of its index:
\begin{align}
    \mathcal{C}_{\ul{k}_x} = \begin{cases} \mathcal{A}_{\ul{k}_x} &\text{if } \ul{k}_x \in \{ \ul{i}_1, \ldots,  \ul{i}_a \} \\ \mathcal{B}_{\ul{k}_x} &\text{if } \ul{k}_x \in \{ \ul{j}_1, \ldots, \ul{j}_b \} \end{cases} \, .
\end{align}
Then assuming that all of the $\ul{k}$ indices are distinct, let
\begin{align}\label{wotimes_defn}
    \mathcal{A}_{\ul{i}_1 \ldots \ul{i}_a} \wotimes \mathcal{B}_{\ul{j}_1 \ldots \ul{j}_b} = \left( - 1 \right)^\sigma \mathcal{C}_{\ul{k}_1} \otimes \ldots \otimes \mathcal{C}_{\ul{k}_{a + b}} \, ,
\end{align}
where $(- 1)^\sigma$ is the signature of the permutation which maps the list of $\ul{i}$ and $\ul{j}$ indices to the list of $\ul{k}$ indices. The definition (\ref{wotimes_defn}) is appropriate if no subsystem index is repeated. Otherwise, we define
\begin{align}\label{wotimes_defn_ZERO}
    \mathcal{A}_{\ul{i}_1 \ldots \ul{i}_a} \wotimes \mathcal{B}_{\ul{j}_1 \ldots \ul{j}_b} = 0 \qquad \text{(if any pair of indices coincide)} \, .
\end{align}
A general operator acting on a multi-partite Hilbert space does not admit a tensor product form (\ref{O_and_tilde_O_factorized}), but is instead a linear combination of such product operators. We extend the definition of $\wotimes$ to such general operators by demanding that this operator be multi-linear, as for the usual tensor product.

In the numerical computations of dimensions of entanglement cohomologies which we performed for this work, e.g. (\ref{example_poincare}), the re-shuffling operation $\wotimes$ was implemented using the \texttt{Permute} method of the \texttt{Qobj} class in QuTiP, the Quantum Toolkit in Python \cite{Johansson:2011jer,Johansson:2012qtx}.

As an example, in the tripartite case, one has
\begin{align}
    \mathcal{O}_A \wotimes \mathcal{O}_C \wotimes \, \mathcal{O}_B = - \mathcal{O}_A \otimes \, \mathcal{O}_B \otimes \mathcal{O}_C \, , 
\end{align}
Using $\wotimes$ resolves the first of the two issues we raised above. The solution to the second issue, as it turns out, is to mimic the action of the product rule by inserting a factor of $2$ by hand, which causes the action of $d^2$ to behave as though it is generating two separate terms. We conclude that the appropriate action of $d^2$ in the tripartite complex is
\begin{align}\label{problematic_fixed}
    d^2 \omega_2 = 2 \left( \mathbb{I}_A \wotimes \mathcal{O}_{BC} + \mathcal{O}_{AC} \wotimes \mathbb{I}_B + \mathcal{O}_{AB} \wotimes \mathbb{I}_C \right) \big\vert_{\rho_{ABC}} \, ,
\end{align}
where we note that $\wotimes$ coincides with the usual tensor product $\otimes$ in the first and third terms of (\ref{problematic_fixed}), but in the second term it rearranges the factors to give an operator that correctly acts on $\mathcal{H}_{ABC}$. 

\subsubsection*{\ul{\it Multipartite complexes}}

Let us now see how the above constructions generalize to an arbitrary multi-partite Hilbert space $\Honeton$. Given a state $\ket{\psi_{\ul{1} \ldots \ul{n}}} \in \Honeton$, and associated density matrix $\rho_{\ul{1} \ldots \ul{n}} = \ket{\psi_{\ul{1} \ldots \ul{n}}} \bra{\psi_{\ul{1} \ldots \ul{n}}}$, the goal is to construct a chain complex of the form
\begin{align}\label{general_npartite_complex}
    0 \xrightarrow{\; \;} \mathbb{C} \xrightarrow{d^0} \Omega^1 \left( \rho_{\ul{1} \ldots \ul{n}} \right) \xrightarrow{d^1} \Omega^2 \left( \rho_{\ul{1} \ldots \ul{n}} \right) \xrightarrow{d^2} \ldots \xrightarrow{d^{n-1}} \Omega^n \left( \rho_{\ul{1} \ldots \ul{n}} \right) \xrightarrow{d^{n}} 0 \, ,
\end{align}
which reduces to (\ref{bipartite_complex_defn}) in the case $n = 2$ and to (\ref{tripartite_complex_defn}) when $n = 3$.

The first coboundary operator will act in the obvious way,
\begin{align}
    d^0 \lambda = \left( \lambda s_{\ul{1}} \, , \, \ldots \, , \, \lambda s_{\ul{n}} \right) \, ,
\end{align}
where $s_{\ul{i}}$ is the support projection onto $\rho_{\ul{i}}$.

Next we define the differential operator $d^m : \Omega^m \left( \rho_{\ul{1} \ldots \ul{n}} \right) \rightarrow \Omega^{m+1} \left( \rho_{\ul{1} \ldots \ul{n}} \right)$ for $1 \leq m < n$. Recall that the elements of $\Omega^{m} \left( \rho_{\ul{1} \ldots \ul{n}} \right)$ are tuples of operators,
\begin{align}
    \omega = \bigtimes_{| \ul{I} | = m } \omega_{\ul{I}} \, ,
\end{align}
acting on reduced Hilbert spaces, where the Cartesian product is taken over all multi-indices associated with length-$m$ subsets of $\{ \ul{1} , \ldots , \ul{n} \}$. Thus the operator $d^m$ should take in a collection of such operators acting on $m$-partite Hilbert spaces, and return a new collection of operators acting on $(m+1)$-partite Hilbert spaces. We do this by defining
\begin{align}\label{general_dm_defn}
    d^m \omega = \bigtimes_{| \ul{I} | = m + 1} \left( m \sum_{\ul{j} \in \ul{I}} \mathbb{I}_{\ul{j}} \wotimes \omega_{\ul{I} \setminus \ul{j}} \right) \, ,
\end{align}
where $\ul{I} \setminus \ul{j}$ denotes the length-$m$ tuple which is obtained by deleting the element $\ul{j}$ from the length-$(m+1)$ tuple $\ul{I}$. As we mentioned before, the factor of $m$ in (\ref{general_dm_defn}) imitates the behavior of the product rule for ordinary derivatives by generating $m$ separate terms.\footnote{We note that this factor does not appear in the definitions of \cite{Mainiero:2019enr}. One can convert between the two conventions for coboundary operators using the dictionary $\left( d^m \right)^{\text{here}} = m \cdot \left( d^{m-1} \right)^{\text{there}}$.}

It is straightforward to check that this coboundary operator correctly reduces to the previously discussed cases for bipartite and tripartite systems. It is also nilpotent, $d^m \circ d^{m-1} = 0$, due to a compatibility of supports argument which generalizes (\ref{compatibility_bipartite}) \cite{Mainiero:2019enr}. One may therefore define the collection of cohomologies $H^m ( \rhooneton )$ according to (\ref{cohomology_defn}).

\subsection{Dimensions of Spaces of Entanglement $k$-forms}

We are primarily interested in the dimensions of \emph{cohomologies}, which are constructed from the various spaces $\Omega^k = \Omega^k ( \rhooneton )$ by considering images and kernels of the coboundary maps (here we suppress the dependence on $\rhooneton$ for brevity). But let us first remark on the dimensions of the vector spaces $\Omega^k$ themselves, before passing to cohomology. In particular, we will compare the dimensions of $\Omega^k$ and $\Omega^{n-k}$ for some fixed $k$ with $1 \leq k < n$. The elements of these two spaces are tuples of the same length, since
\begin{align}
    {n \choose k} = {n \choose n - k} \, .
\end{align}
However, it is not immediately clear that $\dim \left( \Omega^k \right)$ and $\dim \left( \Omega^{n-k} \right)$ coincide as vector spaces; each entry in a tuple defining an element of $\Omega^k$ or $\Omega^{n-k}$ is an operator that has been restricted to the image of a reduced density matrix $\rho_{\ul{I}}$, where $\ul{I}$ is some multi-index. If the image of $\rho_{\ul{I}}$ has dimension $d_{\ul{I}}^{\im}$, then the dimension of the space of linear maps from this image subspace to itself is $\left( d_{\ul{I}}^{\im} \right)^2$. Thus, to compare the dimensions of spaces of entanglement forms, we must account for the dimensions of these image subspaces.

It turns out that a comparison of these images is straightforward, as a consequence of the Schmidt decomposition and the assumption that we always begin with a pure state on the total Hilbert space. We make this comparison precise in the following lemma.

\begin{lemma}\label{dimensions_lemma}
    Consider a density matrix $\rhooneton = \ket{\psioneton} \bra{\psioneton}$ associated with a pure state $\ket{\psioneton} \in \Honeton$. Fix a multi-index $\ul{I}$ which contains a subset of the subsystems $(\ul{1}, \ldots, \ul{n})$ and let $\ul{I}^C$ be the multi-index containing the complement of this subset. Define the reduced density matrices
    \begin{align}
        \rho_{\ul{I}} = \tr_{\ul{I}^C} \left( \rhooneton \right) \, , \qquad \rho_{\ul{I}^C} = \tr_{\ul{I}} \left( \rhooneton \right) \, .
    \end{align}
    Then
    \begin{align}
        \dim \left( \im \left( \rho_{\ul{I}} \right) \right) = \dim \left( \im \left( \rho_{\ul{I}^C} \right) \right) \, .
    \end{align}
\end{lemma}

\begin{proof}
    We may view $\ket{\psioneton}$ as a pure state on a bi-partite system $\left( \ul{I} , \ul{I}^C \right)$ by coarse-graining, that is, by considering $\mathcal{H}_{\ul{I}}$ and $\mathcal{H}_{\ul{I}^C}$ to be two subsystems of the total Hilbert space. Then by the Schmidt decomposition theorem, for some Schmidt rank $S$, there exist orthonormal bases $\ket{ \alpha_{\ul{I}} }$ and $\ket{ \alpha_{\ul{I}^C} }$ for the two subsystem Hilbert spaces such that
    \begin{align}
        \ket{ \psioneton } = \sum_{\alpha = 1}^{S} ( - 1 )^\sigma \lambda_\alpha \ket{ \alpha_{\ul{I}} } \wotimes \ket{ \alpha_{\ul{I}^C} } \, .
    \end{align}
    Here $\sigma \in S_{\ul{n}}$ is the permutation which maps $(\ul{I} , \ul{I}^C)$ to the ordered list $(\ul{1} , \ldots, \ul{n})$; this factor is included to cancel a compensating factor that arises from $\wotimes$, which re-shuffles the tensor products and inserts a similar sign, when $\ul{I}$ and $\ul{I}^C$ are ``out of order''.

    The reduced density matrices associated with the two subsystems are
    \begin{align}
        \rho_{\ul{I}} = \sum_{\alpha = 1}^{S} \lambda_\alpha^2 \ket{\alpha_{\ul{I}}} \bra{\alpha_{\ul{I}}} \, , \qquad \rho_{\ul{I}^C} = \sum_{\alpha = 1}^{S} \lambda_\alpha^2 \ket{\alpha_{\ul{I}^C}} \bra{\alpha_{\ul{I}^C}} \, ,
    \end{align}
    where all $\lambda_\alpha$ are assumed to be non-negative. The corresponding support projections are
    \begin{align}
        s_{\ul{I}} = \sum_{\alpha = 1}^{S} \ket{\alpha_{\ul{I}}} \bra{\alpha_{\ul{I}}} \, , \qquad s_{\ul{I}^C} = \sum_{\alpha = 1}^{S} \ket{\alpha_{\ul{I}^C}} \bra{\alpha_{\ul{I}^C}} \, .
    \end{align}
    In particular, by orthonormality of the two bases, both $s_{\ul{I}}$ and $s_{\ul{I}^C}$ are projectors onto dimension-$S$ subspaces. We conclude that
    \begin{align}
        \dim \left( \im \left( \rho_{\ul{I}} \right) \right) = \dim \left( \im \left( \rho_{\ul{I}^C} \right) \right) \, ,
    \end{align}
    as claimed.
\end{proof}

Having established Lemma \ref{dimensions_lemma}, the desired statement about the dimensions of the spaces of entanglement forms follows as a simple corollary.

\begin{corollary}\label{same_dim_corollary}
    In any entanglement complex 
    \begin{align}
        0 \to \mathbb{C} \xrightarrow{d^0} \Omega^1 \left( \rho_{\ul{1} \ldots \ul{n}} \right) \xrightarrow{d^1} \Omega^2 \left( \rho_{\ul{1} \ldots \ul{n}} \right) \xrightarrow{d^2} \ldots \xrightarrow{d^{n-1}} \Omega^n \left( \rho_{\ul{1} \ldots \ul{n}} \right) \xrightarrow{d^n} 0 \, ,
    \end{align}
    associated with a pure state density matrix $\rhooneton$, one has
    \begin{align}
        \dim \left( \Omega^k \left( \rhooneton \right) \right) = \dim \left( \Omega^{n-k} \left( \rhooneton \right) \right) \, ,
    \end{align}
    for each $k = 1 , \ldots , n-1$.
\end{corollary}

\begin{proof}
    The dimension of $\Omega^k \left( \rhooneton \right)$ is the sum of the dimensions of the spaces of linear operators from $\im \left( \rho_{\ul{I}} \right)$ to $\im \left( \rho_{\ul{I}} \right)$, as $\ul{I}$ runs over all length-$k$ subsets of $(\ul{1}, \ldots, \ul{n})$. By Lemma \ref{dimensions_lemma}, this is equal to the sum of the dimensions of the spaces of linear operators from $\im \left( \rho_{\ul{I}^C} \right)$ to $\im \left( \rho_{\ul{I}^C} \right)$. But the latter sum is precisely the dimension of $\Omega^{n-k} \left( \rhooneton \right)$.
\end{proof}
As any two vector spaces of the same dimension are isomorphic, one concludes that
\begin{align}\label{kform_isomorphism}
    \Omega^k \left( \rhooneton \right) \cong \Omega^{n-k} \left( \rhooneton \right) \, .
\end{align}
However, we reiterate that (\ref{kform_isomorphism}) does not imply that the dimensions of the cohomologies $H^k ( \rhooneton )$ and $H^{n-k} ( \rhooneton )$ are identical, since these spaces are derived from $\Omega^k \left( \rhooneton \right)$ and $\Omega^{n-k} \left( \rhooneton \right)$ by taking appropriate quotients of kernels by images. In order to demonstrate equivalence of the cohomologies, one must find a special isomorphism of the type (\ref{kform_isomorphism}) which enjoys the additional property of interacting with the coboundary operators in a suitable way. This motivates the construction of an explicit Hodge star map between $\Omega^k \left( \rhooneton \right)$ and $\Omega^{n-k} \left( \rhooneton \right)$, which does enjoy such a property.

\section{Hodge Theory for Entanglement}\label{sec:hodge}

As we have just reviewed, entanglement in finite-dimensional quantum systems is characterized by the properties of spaces of objects $\Omega^k ( \rhooneton )$ which resemble differential forms on manifolds. It is natural to wonder whether other constructions for ordinary differential forms also have analogues for entanglement $k$-forms. For instance, an important operation in the study of differential forms on manifolds is the Hodge star:
\begin{align}
    \ast : \Omega^k ( \mathcal{M} ) \to \Omega^{n - k} ( \mathcal{M} ) \, .
\end{align}
In addition to providing a natural duality between $k$-forms and $(n-k)$-forms, the Hodge star allows us to define several related notions, such as an inner product between $k$-forms,
\begin{align}\label{forms_inner_product}
    \langle \alpha , \beta \rangle = \int_{\mathcal{M}} \alpha \wedge \ast \beta\, ,
\end{align}
as well as a codifferential $\delta$, whose action on $k$-forms defined on an $n$-dimensional Riemannian manifold $\mathcal{M}$ is given by
\begin{align}
    \delta = \left( - 1 \right)^{n ( k + 1 ) + 1} \ast d \ast \, .
\end{align}
Whereas the exterior derivative $d$ maps $k$ forms to $(k+1)$-forms, the codifferential $\delta$ sends $k$ forms to $(k-1)$-forms. One can combine these two operations to form the Laplacian,
\begin{align}
    \Delta = \delta d + d \delta \, .
\end{align}
A differential form annihilated by the Laplacian is said to be harmonic; we write $\mathrm{Harm}^k ( \mathcal{M} )$ for the space of harmonic $k$-forms. One of the central results of Hodge theory is that each cohomology class admits a unique harmonic representative, which means that the space of harmonic $k$-forms is isomorphic to the $k$-th cohomology group:
\begin{align}\label{harm_to_coho}
    \mathrm{Harm}^k ( \mathcal{M} ) \cong H^k ( \mathcal{M} ) \, .
\end{align}
A harmonic form is annihilated by both the exterior derivative $d$ and the codifferential $\delta$. Therefore, the Hodge star operation maps harmonic $k$-forms to harmonic $(n-k)$-forms,
\begin{align}\label{hodge_k_to_nmk}
    \ast : \mathrm{Harm}^k ( \mathcal{M} ) \to \mathrm{Harm}^{n-k} ( \mathcal{M} ) \, ,
\end{align}
and in view of the isomorphism (\ref{harm_to_coho}), the map (\ref{hodge_k_to_nmk}) gives rise to an isomorphism between the $k$-th cohomology group and the $(n-k)$-th cohomology group. In particular, the dimensions of $H^k ( \mathcal{M} )$ and $ H^{n-k} ( \mathcal{M} )$ must coincide.

Numerical investigation suggests that the dimensions of entanglement cohomology groups enjoy the same symmetry. It is convenient to collect this numerical data in the form of the Poincar\'e polynomial associated with a state $\rhooneton$, which we define as
\begin{align}\label{poincare_defn}
    P_{n - 2} ( x ) = \sum_{k = 1}^{n - 1} \dim \left( H^k \left( \rhooneton \right) \right) x^{k - 1} \, .
\end{align}
In our conventions (\ref{poincare_defn}) for the Poincar\'e polynomial, the subscript $m$ of $P_m$ does not label the number of subsystems $n$ under consideration, but rather the largest power of $x$ that may appear in the polynomial, which is $m = n - 2$. We have omitted $k = 0$ and $k = n$ from the summation since $H^0$ and $H^n$ are always trivial; the former is obvious, since the kernel of $d^0$ is $0$, and the latter then follows from the Hodge duality we will soon develop. For instance, all bipartite complexes ($n = 2$) have only a single cohomology which may be non-trivial, namely $H^1$, which we have assigned to the constant term of the polynomial. 

Let us illustrate the symmetry of $P_{n-2} ( x )$ by considering the generalized GHZ states
\begin{align}
    \ket{ \mathrm{GHZ}_{n} } = \frac{1}{\sqrt{2}} \left( \underbrace{\ket{0} \otimes \ldots \otimes \ket{0}}_{n \text{ times } } + \underbrace{\ket{1} \otimes \ldots \otimes \ket{1}}_{n \text{ times } } \right) \, ,
\end{align}
The first few instances of the Poincar\'e polynomials $P_{n-2}$ which collect the dimensions of the cohomologies for the states $\rho_{n} = \ket{ \mathrm{GHZ}_{n} } \bra{ \mathrm{GHZ}_{n} }$, are
\begin{align}\label{example_poincare}
    P_{0} &= 6 \, , \nonumber \\
    P_{1} &= 7 + 7 x \, , \nonumber \\
    P_{2} &= 9 + 12 x + 9 x^2 \, , \nonumber \\
    P_{3} &= 11 + 20 x + 20 x^2 + 11 x^3 \, , \nonumber \\
    P_{4} &= 13 + 30 x + 40 x^2 + 30 x^3 + 13 x^4 \, .
\end{align}
In all of the examples (\ref{example_poincare}), the Poincar\'e polynomials are symmetric, in the sense that in each $P_{n-2}$ the coefficient of the term $x^k$ agrees with the coefficient of the term $x^{n-k}$ for all $k$. This pattern persists in all cases which we have investigated numerically.

The goal of this section is to explain this symmetry property by developing an analogue of the isomorphism (\ref{hodge_k_to_nmk}) for entanglement cohomology. Let us point out that, despite the fact that such an isomorphism exists for de Rham cohomology, it is not guaranteed \emph{a priori} that a corresponding isomorphism must exist for entanglement cohomology. As a counter-example, we recall that the original work \cite{Mainiero:2019enr} actually proposed \emph{two} notions of entanglement cohomology: the one which we focus on in this work, which was referred to as the commutant complex in \cite{Mainiero:2019enr}, is based on the restriction defined in (\ref{restriction}), while a variant of this construction called the GNS complex is based on a different restriction
\begin{align}\label{gns_restriction}
    \mathcal{O} \big\|_{\rho} = \mathcal{O} s_\rho \, ,
\end{align}
which multiplies an operator by a support projection on the right but not on the left. The Poincar\'e polynomials associated with the cohomology defined with this alternate restriction (\ref{gns_restriction}) are \emph{not} symmetric, and there is no analogue of a Hodge star operation in this setting. This makes it clear that a version of (\ref{hodge_k_to_nmk}) is not automatic, and that the proof of any Hodge theorem for entanglement cohomology must rely upon the specific structure of the restriction map (\ref{restriction}) which we have used to define our cochain complex.

Let us make one comment on terminology. The analogue of the symmetry of the Poincar\'e polynomials for entanglement cohomology mentioned above, but in the case of the Betti numbers of closed orientable $n$-dimensional manifolds, was first stated by Poincar\'e himself in 1893. In this setting, the statement is simply that the $k$-th and $(n-k)$-th Betti numbers of such a manifold are equal:
\begin{align}\label{symmetry_bettis}
    b_k = b_{n - k} \, .
\end{align}
Some authors refer to a symmetry of the form (\ref{symmetry_bettis}) is as \emph{Poincar\'e duality}. However, in this article we will take a slightly different perspective: we reserve the term ``Poincar\'e duality'' for an equivalence between homology groups and cohomology groups:
\begin{align}
    H^k ( \mathcal{M} ) \cong H_{n - k} ( \mathcal{M} ) \, .
\end{align}
We will also adopt the convention that Betti numbers measure the dimension of \emph{homology} groups, rather than cohomology groups; in most familiar examples, this distinction is immaterial, but there exist spaces for which the dimensions of $H^k$ and $H_k$ disagree. Therefore, using our terminology, it would be inappropriate to refer to the symmetry exhibited by the examples (\ref{example_poincare}) as ``Poincar\'e duality'' as the symmetry involves only the dimensions of \emph{cohomologies}.\footnote{We do not discuss homologies in this paper, although the homology dual to the entanglement cohomology studied here was presented in \cite{Mainiero:2019enr}, and the results of that work imply $\dim ( H_k ) = \dim ( H^k )$.} Because we speak only of cohomology, we will instead refer to this symmetry as ``Hodge duality''.

We begin our exploration of Hodge theory for entanglement by constructing a Hodge star operation which will appear in the definition of an inner product, similar to (\ref{forms_inner_product}), between entanglement $k$-forms. This operation gives rise to an inner product on entanglement $k$-forms, which will illustrate why the restriction (\ref{restriction}) is preferred over (\ref{gns_restriction}).

\subsection{Construction of Hodge Star}\label{sec:hodge_star}

The key element in our development of Hodge theory for entanglement $k$-forms is the construction of a Hodge star map $\ast : \Omega^k ( \rhooneton ) \to \Omega^{n-k} ( \rhooneton )$, which is an example of an isomorphism between these vector spaces whose existence is guaranteed by the argument around equation (\ref{kform_isomorphism}). This construction can be made explicit using a Schmidt decomposition argument similar that of Lemma \ref{dimensions_lemma}. As in that setting, let $\rhooneton = \ket{\psioneton} \bra{\psioneton}$, consider an entanglement $k$-form $\omega \in \Omega^k ( \rhooneton )$, and focus on one component $\alpha_{\ul{I}}$ for a length-$k$ multi-index $\ul{I}$. We perform a Schmidt decomposition
\begin{align}
    \ket{ \psioneton } = \sum_{\alpha = 1}^{S} \lambda_\alpha \ket{ \alpha_{\ul{I}} } \otimes \ket{ \alpha_{\ul{I}^C} } \, ,
\end{align}
so that the support projections associated with the reduced density matrices are
\begin{align}
    s_{\ul{I}} = \sum_{\alpha = 1}^{S} \ket{\alpha_{\ul{I}}} \bra{\alpha_{\ul{I}}} \, , \qquad s_{\ul{I}^C} = \sum_{\alpha = 1}^{S} \ket{\alpha_{\ul{I}^C}} \bra{\alpha_{\ul{I}^C}} \, .
\end{align}
Then $\omega_{\ul{I}}$ is a linear operator on the $S$-dimensional vector space spanned by the $\ket{\alpha_{\ul{I}} }$,
\begin{align}
    \omega_{\ul{I}} = \sum_{\alpha, \beta = 1}^{S} \left( \omega_{\ul{I}} \right)_{\alpha \beta} \ket{ \alpha_{\ul{I}}} \bra{\beta_{\ul{I}}} \, .
\end{align}
We define a corresponding linear operator acting on the $S$-dimensional vector space spanned by the orthonormal basis vectors $\ket{\alpha_{\ul{I}^C}}$ by
\begin{align}\label{hodge_defn}
    \left( \ast \omega \right)_{\ul{I}^C} = \sum_{\alpha, \beta = 1}^{S} \left( - 1 \right)^\sigma \left( \omega_{\ul{I}} \right)_{\alpha \beta}^\ast \ket{ \alpha_{\ul{I}^C}} \bra{\beta_{\ul{I}^C}} \, ,
\end{align}
where $\sigma \in S_{\ul{n}}$ is the permutation which sends the ordered tuple $\ul{I} \cup \ul{I}^C$ to $( \ul{1} , \ldots , \ul{n} )$, and $\left( \omega_{\ul{I}} \right)_{\alpha \beta}^\ast$ is the complex conjugate of the matrix element $\left( \omega_{\ul{I}} \right)_{\alpha \beta}$. Collecting the operators (\ref{hodge_defn}) for all multi-indices $\ul{I}$ defines an entanglement $(n-k)$ form $\ast \omega$.

Let us now explain the reason for taking the complex-conjugate of the elements of $\omega$, which is related to the observation that the Schmidt decomposition is not unique. In particular, one is free to rotate the bases $\ket{ \alpha_{\ul{I} } }$ and $\ket{\alpha_{\ul{I}^C}}$ by compensating phases in a way which leaves the decomposition unchanged. For instance, one may redefine
\begin{align}
    \ket{ \alpha_{\ul{I} } } \to \ket{ \widetilde{\alpha}_{\ul{I} } }  = e^{i \theta_\alpha} \ket{ \alpha_{\ul{I} } } \, , \qquad \ket{\alpha_{\ul{I}^C}} \to \ket{\widetilde{\alpha}_{\ul{I}^C}} = e^{- i \theta_\alpha} \ket{\alpha_{\ul{I}^C}} \, .
\end{align}
Under such a phase rotation of the bases, we have
\begin{align}
    \omega_{\ul{I}} = \sum_{\alpha, \beta = 1}^{S} \left( \omega_{\ul{I}} \right)_{\alpha \beta} e^{- i \left( \theta_\alpha - \theta_\beta \right) }  \ket{ \widetilde{\alpha}_{\ul{I}}} \langle \widetilde{\beta}_{\ul{I}} \vert \, ,
\end{align}
so the matrix elements of $\omega_{\ul{I}}$ in the rotated basis are
\begin{align}
    \left( \widetilde{\omega}_{\ul{I}} \right)_{\alpha \beta} = \left( \omega_{\ul{I}} \right)_{\alpha \beta} e^{- i \left( \theta_\alpha - \theta_\beta \right) } \, .
\end{align}
Taking the Hodge dual in the new rotated basis, using the definition (\ref{hodge_defn}) which includes the complex conjugation, gives
\begin{align}\label{hodge_defn_rotated}
    \left( \ast \omega \right)_{\ul{I}^C} &= \sum_{\alpha, \beta = 1}^{S} \left( - 1 \right)^\sigma \left( \widetilde{\omega}_{\ul{I}} \right)_{\alpha \beta}^\ast \ket{ \widetilde{\alpha}_{\ul{I}^C}} \big\langle \, \widetilde{\beta}_{\ul{I}^C} \big\vert \nonumber \\
    &= \sum_{\alpha, \beta = 1}^{S} \left( - 1 \right)^\sigma \left( \omega_{\ul{I}} \right)_{\alpha \beta}^\ast e^{i \left( \theta_\alpha - \theta_\beta \right) } \ket{ \widetilde{\alpha}_{\ul{I}^C}} \big\langle \, \widetilde{\beta}_{\ul{I}^C} \big\vert \nonumber \\
    &= \sum_{\alpha, \beta = 1}^{S} \left( - 1 \right)^\sigma \left( \omega_{\ul{I}} \right)_{\alpha \beta}^\ast \ket{ \alpha_{\ul{I}^C}} \bra{\beta_{\ul{I}^C}} \, .
\end{align}
The last line of (\ref{hodge_defn_rotated}) is identical to the expression (\ref{hodge_defn}) using the original basis. Therefore, by virtue of performing the complex conjugation of the matrix elements $\left( \omega_{\ul{I}} \right)_{\alpha \beta}$, our definition of the Hodge star is invariant under the phase ambiguity of the Schmidt bases.

Finally, given this definition of the Hodge star, one can show that
\begin{align}
    \ast \ast = ( -1 )^{k ( n - k )} \, ,
\end{align}
when acting on elements of $\Omega^k ( \rhooneton )$.

\subsection{Inner Product for Entanglement $k$-forms}

Fix a density matrix $\rhooneton$ and let $\omega, \eta \in \Omega^k ( \rhooneton )$. We would like to follow the recipe (\ref{forms_inner_product}) for defining an inner product between differential $k$-forms in order to find an analogous inner product $\langle \omega , \eta \rangle$. As a first step, we define a wedge product on entanglement forms
\begin{align}
    \omega \wedge \eta = \sum_{\ul{I} , \ul{J}} \left( \omega_{\ul{I}} \wotimes \eta_{\ul{J}} \right) \big\vert_{\rho_{\ul{I} \cup \ul{J} }} \, ,
\end{align}
where the sum runs over all multi-indices labeling the components of $\omega$ and $\eta$, and in each term we project onto the image of the density matrix associated to the Hilbert space which includes all of the subsystems in both $\ul{I}$ and $\ul{J}$.

Note that, according to our definition of $\wotimes$, all terms in this wedge product which involve pairs of multi-indices that share a common index will vanish, as with ordinary differential forms on manifolds. As a simple example, given two entanglement $1$-forms $\omega = ( \omega_A , \omega_B )$ and $\eta = ( \eta_A , \eta_B )$ on a bipartite Hilbert space $\mathcal{H}_{AB}$, one finds
\begin{align}
    \omega \wedge \eta = \left( \omega_A \otimes \eta_B - \eta_A \otimes \omega_B \right) \big\vert_{\rho_{AB}} \, .
\end{align}
It is straightforward to check that this operation is associative,
\begin{align}
    \omega \wedge \left( \eta \wedge \xi \right) = \left( \omega \wedge \eta \right) \wedge \xi \, ,
\end{align}
and by virtue of the factor of $m$ which we included in the definition (\ref{general_dm_defn}) of the differential, it satisfies the Leibniz formula
\begin{align}
    d \left( \omega \wedge \eta \right) = ( d \omega ) \wedge \eta + ( - 1 )^p \omega \wedge ( d \eta ) \, ,
\end{align}
where $\omega$ is an entanglement $p$-form and $\eta$ is an entanglement $q$-form. In particular, this endows the spaces $\Omega^k ( \rhooneton )$ with the structure of a differential graded algebra.\footnote{We are grateful to Tom Mainiero for helpful discussions on this point.}

Given this wedge product, we define an inner product on entanglement $k$-forms by
\begin{align}\label{inner_product_to_hodge}
    \langle \omega, \eta \rangle = \tr \left( \left( \omega \wedge \ast \eta \right) \big\vert_{\rhooneton} \right) \, ,
\end{align}
where the trace in (\ref{inner_product_to_hodge}) is taken in the full Hilbert space $\Honeton$. Since the definition (\ref{hodge_defn}) of the Hodge star involves complex conjugation of the matrix elements, this inner product is linear in the first argument and antilinear in the second argument,
\begin{align}
    \langle z \omega, w \eta \rangle = z \overbar{w} \langle \omega , \eta \rangle \, ,
\end{align}
which is the opposite of the usual convention for a sesquilinear form. However, this does not substantially affect any of its properties; one could instead define the inner product as the complex conjugate of (\ref{inner_product_to_hodge}), which would be antilinear in the first argument and linear in the second, without changing any of our conclusions.

To elucidate the properties of this inner product, it will be useful to develop an explicit formula for it, again relying on the Schmidt decomposition. Recall that $\omega$ and $\eta$ are tuples of operators acting on subsystems of the total Hilbert space $\Honeton$ consisting of $k$ tensor product factors, which we write schematically as 
\begin{align}
    \omega = \left( \omega_{ \underline{I}_1 } , \ldots , \omega_{ \underline{I}_N } \right) \, , \qquad \eta = \left( \eta_{ \underline{I}_1 } , \ldots , \eta_{ \underline{I}_N } \right) \, ,
\end{align}
where $N = {n \choose k}$, and we assume that each component of $\omega$ and $\eta$ are suitably restricted:
\begin{align}
    \omega_{\ul{I}} = \left( \omega_{\ul{I}} \right) \big\vert_{\rho_{\ul{I}}} \, , \qquad \eta_{\ul{I}} = \left( \eta_{\ul{I}} \right) \big\vert_{\rho_{\ul{I}}} \, .
\end{align}
As $\omega$ is an entanglement $k$-form and $\eta$ is an entanglement $(n - k)$-form, the only non-vanishing contributions to their wedge product come from pairing a component $(\ast \omega)_{\ul{I}^C}$ with a complementary component $\eta_{\ul{I}}$ for some multi-index $\ul{I}$. We will compute the contribution from each complementary pair separately, and then add the results. For a given partition of the subsystems $\ul{1} , \ldots , \ul{n}$ into $\ul{I}$ and $\ul{I}^C$, we perform a Schmidt decomposition of the pure state $\ket{\psioneton}$ defining the density matrix $\rhooneton$ on the total Hilbert space as
\begin{align}
    \ket{\psioneton} = \sum_{\alpha = 1}^{S} ( - 1 )^\sigma \lambda_\alpha \ket{\alpha_{\ul{I}}} \wotimes \ket{\alpha_{\ul{I}^C}} \, .
\end{align}
Again $\sigma$ is the permutation which maps the ordered list $(\ul{I} , \ul{I}^C)$ to $\ul{1} , \ldots , \ul{n}$, which we include only to compensate the corresponding factor arising from $\wotimes$, which puts the tensor product of basis kets ``back in order'' to define a valid state on $\Honeton$.

The density matrix $\rhooneton$ is a rank-one projector, as it is a pure state density matrix, so it is equal to its own support projection,
\begin{align}\label{schmidt_for_inner_product}
    \soneton = \rhooneton = \sum_{\alpha, \beta = 1}^{S} \lambda_\alpha \lambda_\beta \left( \ket{\alpha_{\ul{I}}} \wotimes \ket{\alpha_{\ul{I}^C}} \right) \left( \bra{\beta_{\ul{I}}} \wotimes \bra{\beta_{\ul{I}^C}} \right) \, ,
\end{align}
while the reduced density matrices and their support projections are
\begin{gather}
    \rho_{\ul{I}} = \sum_{\alpha = 1}^{S} \lambda_\alpha^2 \ket{\alpha_{\ul{I}}} \bra{\alpha_{\ul{I}}} \, , \qquad s_{\ul{I}} = \sum_{\alpha = 1}^{S}  \ket{\alpha_{\ul{I}}} \bra{\alpha_{\ul{I}}} \, , \nonumber \\
    \rho_{\ul{I}^C} = \sum_{\alpha = 1}^{S} \lambda_\alpha^2 \ket{\alpha_{\ul{I}^C}} \bra{\alpha_{\ul{I}^C}} \, , \qquad s_{\ul{I}^C} = \sum_{\alpha = 1}^{S}  \ket{\alpha_{\ul{I}^C}} \bra{\alpha_{\ul{I}^C}} \, .
\end{gather}
Now we consider the combination
\begin{align}\label{omega_wedge_ast_eta}
    \omega \wedge \ast \eta = \sum_{\ul{I}} \sum_{\alpha, \beta, \gamma , \delta = 1}^{S} \left( \left( \omega_{\ul{I}} \right)_{\alpha \beta} \ket{\alpha_{\ul{I}}} \bra{\beta_{\ul{I}}} \right) \wotimes \left(  \left( \eta^\ast_{\ul{I}} \right)_{\gamma \delta} \ket{ \gamma_{\ul{I}^C } } \bra{\delta_{\ul{I}^C}} \right) \, .
\end{align}
We must now project the combination (\ref{omega_wedge_ast_eta}) onto the image of the total density matrix $\rhooneton$. This involves left- and right-multiplying by the support projection $\soneton$, which gives
\begin{align}
    \omega \wedge \ast \eta \big\vert_{\rhooneton} &= \soneton \left( \omega \wedge \ast \eta \right) \soneton \nonumber \\
    &= \left( \sum_{\walpha, \wbeta = 1}^{S} \lambda_\walpha \lambda_\wbeta \left( \ket{\walpha_{\ul{I}}} \wotimes \ket{\walpha_{\ul{I}^C}} \right) \left( \big\langle \, \wbeta_{\ul{I}} \, \big\vert \wotimes \big\langle \, \wbeta_{\ul{I}^C} \, \big\vert \right) \right) \nonumber \\
    &\qquad \cdot \left( \sum_{\ul{I}} \sum_{\alpha, \beta, \gamma , \delta = 1}^{S} \left( \left( \omega_{\ul{I}} \right)_{\alpha \beta} \ket{\alpha_{\ul{I}}} \bra{\beta_{\ul{I}}} \right) \wotimes \left(  \left( \eta^\ast_{\ul{I}} \right)_{\gamma \delta} \ket{ \gamma_{\ul{I}^C } } \bra{\delta_{\ul{I}^C}} \right)  \right) \nonumber \\
    &\qquad \cdot \left( \sum_{\wgamma , \wdelta = 1}^{S} \lambda_\wgamma \lambda_\wdelta \left( \ket{\wgamma_{\ul{I}}} \wotimes \ket{\wgamma_{\ul{I}^C}} \right) \left( \big\langle \, \wdelta_{\ul{I}} \big\vert \wotimes \big\langle \, \wdelta_{\ul{I}^C} \, \big\vert \right) \right) \nonumber \\
    &= \sum_{\ul{I}} \sum_{\walpha, \wbeta = 1}^{S} \sum_{\alpha, \beta, \gamma , \delta = 1}^{S} \sum_{\wgamma , \wdelta = 1}^{S} \Bigg( \lambda_\walpha \lambda_\wbeta  \lambda_\wgamma \lambda_\wdelta \left( \omega_{\ul{I}} \right)_{\alpha \beta} \left( \eta^\ast_{\ul{I}} \right)_{\gamma \delta} \delta_{\widetilde{\beta} \alpha} \delta_{\widetilde{\beta} \gamma} \delta_{\beta \wgamma} \delta_{\delta \wgamma}  \nonumber \\
    &\hspace{140pt} \cdot  \left( \ket{\walpha_{\ul{I}}} \wotimes \ket{\walpha_{\ul{I}^C}} \right) \left( \big\langle \, \wdelta_{\ul{I}} \big\vert \wotimes \big\langle \, \wdelta_{\ul{I}^C} \, \big\vert \right) \Bigg) \nonumber \\
    &= \sum_{\ul{I}} \soneton \lambda_{\wbeta} \lambda_{\wgamma} \left( \left( \omega_{\ul{I}} \right)_{\wbeta \wgamma} \left( \eta_{\ul{I}^\ast} \right)_{\wbeta \wgamma} \right) \, ,
\end{align}
where in the last step we have recognized the definition of the support projection $\soneton$. Relabeling indices for clarity, and evaluating $\tr ( \soneton ) = 1$, we conclude
\begin{align}\label{inner_product_components}
    \langle \omega, \eta \rangle = \tr \left( \left( \omega \wedge \ast \eta \right) \big\vert_{\rhooneton} \right) = \sum_{\ul{I}} \sum_{\alpha, \beta = 1}^{S} \lambda_\alpha \lambda_\beta \left( \omega_{\ul{I}} \right)_{\alpha \beta} \left( \eta^\ast_{\ul{I}} \right)_{\alpha \beta} \, ,
\end{align}
where $\left( \omega_{\ul{I}} \right)_{\alpha \beta}$ and $\left( \eta_{\ul{I}} \right)_{\alpha \beta}$ are the matrix elements of $\omega_{\ul{I}}$ and $\eta_{\ul{I}}$, respectively, in the Schmidt basis $\ket{\alpha_{\ul{I}}}$ associated with the subsystem $\ul{I}$, and $\lambda_\alpha$ are the corresponding Schmidt coefficients, according to the decomposition (\ref{schmidt_for_inner_product}).

It is instructive to study the special case of this inner product when all Schmidt coefficients are equal, $\lambda_\alpha = \lambda$, which is called a \emph{flat entanglement spectrum}. Then
\begin{align}\label{inner_product_same_lambdas}
    \langle \omega, \eta \rangle &= \sum_{\ul{I}} \lambda^2 \sum_{\alpha, \beta = 1}^{S} \left( \omega_{\ul{I}} \right)_{\alpha \beta} \left( \eta^\ast_{\ul{I}} \right)_{\alpha \beta} \nonumber \\
    &= \lambda^2 \sum_{\ul{I}} \left( \omega_{\ul{I}} \right)_{\alpha \beta} \left( \eta^\dagger_{\ul{I}} \right)_{\beta \alpha} \nonumber \\
    &= \lambda^2 \sum_{\ul{I}} \tr \left( \omega_{\ul{I}} \eta^\dagger_{\ul{I}} \right) \, .
\end{align}
In this case, the inner product is proportional to a sum of the standard\footnote{It is more common to define the inner product between operators as $\langle \mathcal{O}_1 , \mathcal{O}_2 \rangle = \tr \left( \mathcal{O}_1^\dagger \mathcal{O}_2 \right)$, but as we commented above, our inner product is antilinear in the second slot rather than in the first.} inner products
\begin{align}\label{usual_inner_product}
    \langle \mathcal{O}_1 , \mathcal{O}_2 \rangle = \tr \left( \mathcal{O}_1 \mathcal{O}_2^\dagger \right) \, ,
\end{align}
on each of the restricted Hilbert spaces $\mathrm{im} ( \rho_{\ul{I}} ) $. However, in the general case of unequal Schmidt coefficients, our inner product (\ref{inner_product_components}) is a generalization of the usual structure (\ref{usual_inner_product}) which includes weightings by the Schmidt coefficients.

From the formula (\ref{inner_product_components}) it is also clear that this inner product is positive-definite, as
\begin{align}\label{inner_product_pos_def}
    \langle \omega, \omega \rangle = \sum_{\ul{I}} \sum_{\alpha, \beta = 1}^{S} \lambda_\alpha \lambda_\beta \left| \left( \omega_{\ul{I}} \right)_{\alpha \beta} \right|^2 \, ,
\end{align}
which is a sum of positive-definite quantities, since the Schmidt coefficients $\lambda_\alpha$ are all positive by convention.\footnote{Here by ``positive-definite'' we mean $\langle \omega , \omega \rangle \geq 0$, with equality if and only if $\omega = 0$.}

The existence of the inner product (\ref{inner_product_to_hodge}) endows each of the vector spaces $\Omega^k ( \rhooneton )$ with the structure of an inner product space. In particular, this means that each $\Omega^k ( \rhooneton )$ is a Hilbert space, since every finite-dimensional inner product space is Hilbert.

Given such an inner product, we may define a codifferential operator $\delta$ which satisfies
\begin{align}\label{codiff_defn}
    \langle \omega, d \eta \rangle = \langle \delta \omega, \eta \rangle \, .
\end{align}
That is, $\delta$ is the unique linear operator which is the adjoint of the coboundary operator $d$. A form which is annihilated by $\delta$ is said to be \emph{co-closed}.

The existence of such an operator $\delta$ is guaranteed since, in finite-dimensional Hilbert spaces, adjoint operators always exist; to see this, note that one can express $d$ in components, as a matrix defined with respect to an orthonormal basis for the Hilbert space $\Omega^k ( \rhooneton )$, and then take the conjugate-transpose of this matrix. Since $dd = 0$, one has
\begin{align}
    0 = \langle \omega , d d \eta \rangle = \langle \delta \delta \omega , \eta \rangle \, ,
\end{align}
so $\delta \delta = 0$, and thus the codifferential inherits the nilpotency of the coboundary operator.

When acting on elements of $\Omega^k ( \rhooneton )$, one can show by explicit calculation that the codifferential $\delta$, defined by (\ref{codiff_defn}), is proportional to $\ast d \ast$:
\begin{align}\label{delta_to_stardstar}
    \delta = c_{n, k} \left( - 1 \right)^{n ( k + 1 ) } \ast d \ast \, .
\end{align}
Here $c_{n, k}$ is a constant whose value depends on $n$ and $k$, which arises from the fact that we have introduced a factor of $m$ in the definition of the coboundary operator (\ref{general_dm_defn}) only for $m \geq 1$. Explicitly, the values of these constants are
\begin{align}\label{cnk_defn}
    c_{n, k} = \begin{cases} \frac{1}{n - 1} &\text{if } k = 1 \\ \frac{k - 1}{ n - k } &\text{if } 1 < k < n \\ n - 1 &\text{if } k = n \end{cases} \, .
\end{align}

\subsection{Hodge Isomorphism, Decomposition, and Duality}

In this section, we aim to prove versions of the Hodge isomorphism theorem, Hodge decomposition, and Hodge duality for entanglement cohomology. In fact, we will prove the first two of these results in a slightly more general setting. We commented, below equation (\ref{inner_product_pos_def}), that our inner product endows each of the vector spaces $\Omega^k ( \rhooneton )$ with the structure of a finite-dimensional Hilbert space. We will therefore give the proofs of the Hodge isomorphism theorem and Hodge decomposition for general cochain complexes, whose associated vector spaces are finite-dimensional Hilbert spaces, and then specialize to the case of entanglement cohomologies at the end.

Let us begin with the following result on harmonic representatives.

\begin{lemma}\label{lemma_harmonic_rep}
Consider a cochain complex
\begin{align}
    \cdots \xrightarrow{d^{-1}} V_0 \xrightarrow{d^{0}} V_1 \xrightarrow{d^{1}} V_2 \xrightarrow{d^{2}} \ldots \, ,
\end{align}
in which each of the $V_i$ is a finite-dimensional Hilbert space with inner product $\langle \cdot \, , \, \cdot \rangle : V_i \times V_i \to \mathbb{C}$. Let $\delta^i : V_{i + 1} \to V_{i}$ be the $i$-th codifferential, which is the adjoint of the coboundary operator $d^i$ with respect to this inner product. Then in each cohomology class there is a unique element which is annihilated by $\delta$, and which is the representative with the smallest norm in its cohomology class. 

\end{lemma}

\begin{proof}
    One trivially has the orthogonal decomposition
    \begin{align}\label{Vi_decomp}
        V_i = \mathrm{im} \left( d^{i-1} \right) \oplus \left( \mathrm{im} \left( d^{i-1} \right) \right)^\perp \, ,
    \end{align}
    where $\perp$ denotes the orthogonal complement with respect to the inner product. But 
    \begin{align}
        \left( \mathrm{im} \left( d^{i-1} \right) \right)^\perp = \ker \left( \delta^{i-1} \right) \, ,
    \end{align}
    in finite-dimensional Hilbert spaces, since
    \begin{align}
        \langle d^{i-1} \omega , \eta \rangle = \langle \omega, \delta^{i - 1} \eta \rangle \, ,
    \end{align}
    and thus $\eta \in V_i$ is orthogonal to the image of $d^{i-1}$ if and only if it is annihilated by $\delta^{i - 1}$.

    Therefore the decomposition (\ref{Vi_decomp}) can be written as
    \begin{align}\label{Vi_decomp_two}
        V_i = \im \left( d^{i-1} \right) \oplus \ker \left( \delta^{i - 1} \right) \, .
    \end{align}
    Let $[\omega]$ be a cohomology class in $V_i$ with representative $\omega$, so that $d^i \omega = 0$. We decompose $\omega \in V_i$ according to (\ref{Vi_decomp_two}) as
    \begin{align}
        \omega = d^{i-1} \lambda + \eta \, ,
    \end{align}
    where $\lambda \in V_{i-1}$ and $\delta^{i-1} \eta = 0$. Note that, since
    \begin{align}
        d^i \omega = d^i \left( d^{i-1} \lambda + \eta \right) = d^i \eta = 0 \, ,
    \end{align}
    that $\eta$ is annihilated by both $d^i$ and $\delta^{i-1}$. This choice of $\eta$ is unique, by the uniqueness of orthogonal decompositions; said differently, any other $\eta' \neq \eta$ in the class $[\omega]$ will differ from $\eta$ by an exact form, which means that it would have a non-trivial projection onto $\im ( d^{i-1} )$, contradicting the orthogonal decomposition (\ref{Vi_decomp_two}).

    Finally, $\eta$ has the smallest norm in the cohomology class $[\omega]$ since for any $\lambda \in V_{i-1}$,
    \begin{align}
        \langle \eta + d^{i-1} \lambda , \eta + d^{i-1} \lambda \rangle &= \langle \eta , \eta \rangle + 2 \langle \eta, d^{i-1} \lambda \rangle + \langle d^{i-1} \lambda , d^{i-1} \lambda \rangle \\
        &= \langle \eta , \eta \rangle + 2 \langle \delta^{i-1} \eta, \lambda \rangle + \langle d^{i-1} \lambda , d^{i-1} \lambda \rangle \nonumber \\
        &= \langle \eta , \eta \rangle + \langle d^{i-1} \lambda , d^{i-1} \lambda \rangle \nonumber \\
        &\geq \langle \eta , \eta \rangle  \, . \qedhere
    \end{align}
\end{proof}
We refer to any element $\eta$ which is annihilated by both the appropriate coboundary operator $d$ and codifferential $\delta$ as a \emph{harmonic form}, and if $\eta$ belongs to a cohomology class $[\omega]$, we say that $\eta$ is a \emph{harmonic representative} of $[ \omega ]$. The content of Proposition \ref{lemma_harmonic_rep}, when applied to the special case of any entanglement complex (\ref{general_npartite_complex}) with inner product (\ref{inner_product_to_hodge}), is that every entanglement cohomology class has a unique harmonic representative.

As we mentioned above, we can apply Proposition \ref{lemma_harmonic_rep} to the setting of entanglement cohomology due to the result that each $\Omega^k ( \rhooneton )$ is a finite-dimensional Hilbert space. Of course the analogous statement about the existence of harmonic representatives also holds for differential forms -- this is the usual Hodge theorem -- but the proof is more involved, since the space $\Omega^k ( \mathcal{M} )$ of differential $k$-forms on a manifold $\mathcal{M}$ is infinite-dimensional, and does not form a Hilbert space because the inner product (\ref{forms_inner_product}) is not complete.

Having shown, in Proposition \ref{lemma_harmonic_rep}, that each cohomology class has a unique harmonic representative, we now show the converse: each harmonic form can be uniquely associated to a cohomology class. This establishes an isomorphism between $\mathrm{Harm}^k$, the space of harmonic $k$-forms in a complex, and $H^k$, the $k$-th cohomology group of the complex.

\begin{thm}[Hodge isomorphism]\label{bijection}
Given a cochain complex satisfying the assumptions of Proposition \ref{lemma_harmonic_rep}, let $\mathrm{Harm}^k$ denote the vector space of harmonic elements of $V_k$. Then
\begin{align}
    \mathrm{Harm}^k \cong H^k \, .
\end{align}
\end{thm}

\begin{proof}
    Every harmonic form is closed, so we may define the map
    \begin{align}
        f : \mathrm{Harm}^k \to H^k \, ,
    \end{align}
    which sends a harmonic form $\omega \in \mathrm{Harm}^k$ to the cohomology class $[\omega] \in H^k$.
    
    This is a surjection, since by Proposition \ref{lemma_harmonic_rep}, every cohomology class has a harmonic representative. To show that it is an injection,  suppose that $f ( \omega ) = f ( \eta )$, or $[ \omega ] = [ \eta ]$. This means that the harmonic forms $\omega$ and $\eta$ are in the same cohomology class, so
    \begin{align}
        \omega - \eta = d^{k-1} \lambda \, ,
    \end{align}
    for some $\lambda \in V_{k-1}$. But then
    \begin{align}
        \langle\omega-\eta , \omega-\eta\rangle &= \langle d^{k-1} \lambda, \omega-\eta\rangle \nonumber \\
        & =\langle\lambda, \delta^{k-1} (\omega-\eta)\rangle \nonumber \\
        & =\langle\lambda , 0\rangle \nonumber \\
        &=0 \, ,
    \end{align}
    where we used that $\omega$, $\eta$ are harmonic, and thus are annihilated by $\delta^{k-1}$. By the positive-definiteness of the inner product, this implies that $\omega = \eta$, which establishes injectivity.

    Therefore, the map $f$ is a bijection. It is also clearly linear. We therefore conclude that $f$ is an isomorphism between the vector spaces $\mathrm{Harm}^k$ and $H^k$.    
\end{proof}
A variant of the Hodge isomorphism theory for differential forms, which generalizes the Helmholtz decomposition, is the Hodge decomposition
\begin{align}
    \omega = d \lambda + \delta \eta + \xi \, ,
\end{align}
which states that any differential $k$-form can be uniquely decomposed into the sum of a closed form, a co-closed form, and a harmonic form. We now prove that an identical statement holds for complexes of finite-dimensional Hilbert spaces, and therefore (as a special case) for entanglement cohomology.

\begin{thm}[Hodge decomposition]\label{hodge_decomp}
Consider a cochain complex satisfying the assumptions of Proposition \ref{lemma_harmonic_rep} and let $\omega \in V_i$. Then $\omega$ admits a unique decomposition
\begin{align}\label{hodge_decomposition}
    \omega = d^{i-1} \lambda + \delta^i \eta + \xi \, ,
\end{align}
for forms $\lambda \in V_{i-1}$, $\eta \in V_{i+1}$, and $\xi \in \mathrm{Harm}^i$.
\end{thm}

\begin{proof}
    We first use the orthogonal decomposition (\ref{Vi_decomp_two}) to write
    \begin{align}
        \omega = d^{i-1} \lambda + \widetilde{\eta} \, ,
    \end{align}
    where $\widetilde{\eta}$ is co-closed, but because $\omega$ is not assumed to be closed, $\widetilde{\eta}$ need not be closed.
    
    We now further decompose $\widetilde{\eta}$ according to the orthogonal decomposition
    \begin{align}
        V_i &= \im \left( \delta^{i} \right) \oplus \left( \im \left( \delta^i \right) \right)^\perp \nonumber \\
        &= \im \left( \delta^i \right) \oplus \ker \left( d^i \right) \, ,
    \end{align}
    which allows us to write
    \begin{align}
        \widetilde{\eta} = \delta^i \eta + \xi \, .
    \end{align}
    Then $d^i \xi = \delta^{i-1} \xi = 0$, so $\xi$ is harmonic. Combining these decompositions gives
    \begin{align}
        \omega = d^{i-1} \lambda + \delta^i \eta + \xi \, ,
    \end{align}
    which is again unique because of the uniqueness of orthogonal decompositions.
\end{proof}

The above results hold for generic cochain complexes whose associated vector spaces are finite-dimensional Hilbert spaces. We now specialize to the case of entanglement cohomology, which affords us the additional structure that the coboundary operation $\delta$ may be expressed in terms of a Hodge star operator as in equation (\ref{delta_to_stardstar}). This assumption makes the entanglement complex behave similarly to the de Rham complex, rather than merely an abstract differential complex with inner products.

Let us take this opportunity to introduce some additional notation, for completeness. We define the Laplacian $\Delta : \Omega^k ( \rhooneton ) \to \Omega^k ( \rhooneton )$ as
\begin{align}
    \Delta = d \delta + \delta d \, .
\end{align}
For the remainder of this section, we will suppress the indices $d^{i}$, $\delta^{i}$, etc. on coboundary and codifferential operators, for simplicity.

We previously defined harmonic forms as those which are annihilated by both $d$ and $\delta$; an equivalent definition is that harmonic forms are annihilated by the Laplacian $\Delta$. One direction of this equivalence is clear, since if $d \omega = 0 = \delta \omega$, then $\Delta \omega = 0$. To see the opposite implication, note that if
\begin{align}
    \Delta \omega = d \delta \omega + \delta d \omega = 0 \, ,
\end{align}
then if we perform a Hodge decomposition of $\omega$ as in equation (\ref{hodge_decomposition}),
\begin{align}
    \omega = d \lambda + \delta \eta + \xi \, ,
\end{align}
then one finds
\begin{align}\label{harmonic_condition}
    0 &= d \delta \left( d \lambda + \delta \eta + \xi \right) + \delta d \left( d \lambda + \delta \eta + \xi \right) \nonumber \\
    &= d \delta d \lambda + \delta d \delta \eta \, .
\end{align}
Then one has
\begin{align}
    \langle d \omega, d \omega \rangle &= \langle d \delta \eta , d \delta \eta \rangle \nonumber \\
    &= \langle \delta \eta, \delta d \delta \eta \rangle \nonumber \\
    &= - \langle \delta \eta, d \delta d \lambda \rangle \nonumber \\
    &= - \langle \delta \delta \eta, \delta d \lambda \rangle \nonumber \\
    &= 0 \, ,
\end{align}
where we used (\ref{harmonic_condition}) and $\delta \delta = 0$, which implies that $d \omega = 0$ by positive-definiteness of the inner product. Likewise,
\begin{align}
    \langle \delta \omega, \delta \omega \rangle &= \langle \delta d \lambda, \delta d \lambda \rangle \nonumber \\
    &= \langle d \lambda, d \delta d \lambda \rangle \nonumber \\
    &= - \langle d \lambda, \delta d \delta \eta \rangle \nonumber \\
    &= - \langle d d \lambda, d \delta \eta \rangle \nonumber \\
    &= 0 \, ,
\end{align}
where we again used (\ref{harmonic_condition}) and $dd = 0$. We conclude that
\begin{align}
    \Delta \omega = 0 \, \iff \, d \omega = 0 = \delta \omega \, ,
\end{align}
so one may characterize harmonic forms as either those forms which are both closed and co-closed, or as those which are annihilated by the Laplacian, as claimed.

The advantage of restricting to entanglement cohomology, where the codifferential $\delta$ is related to $d$ and the Hodge star, is that in this setting the Hodge star operation maps harmonic forms to harmonic forms. We now turn to explaining this fact.

\begin{thm}[Hodge duality]\label{hodge_duality}

Consider a pure state density matrix $\rhooneton$ and the spaces $\Omega^k ( \rhooneton )$ of the associated entanglement $k$-forms. Let $\mathrm{Harm}^k ( \rhooneton ) \subset \Omega^k ( \rhooneton )$ be the space of harmonic $k$-forms, i.e. those that are annihilated by both $d$ and $\delta$. Then
\begin{align}
    \mathrm{Harm}^k ( \rhooneton ) \cong \mathrm{Harm}^{n-k} ( \rhooneton ) \, .
\end{align}
\end{thm}

\begin{proof}

Suppose $\omega \in \mathrm{Harm}^k ( \rhooneton )$ so that $d \omega = 0 = \delta \omega$. Consider the Hodge dual $\ast \omega$ of this form. Since $\delta = ( - 1 )^{n ( k + 1 ) } c_{n, k} \ast d \ast$, in the notation of (\ref{cnk_defn}), one has
\begin{align}
    d \ast \omega = ( - 1)^{k ( n - k ) } \ast \left( \ast d \ast \omega \right) = \frac{1}{c_{n, k}} ( - 1)^{k ( n - k ) } ( - 1 )^{n ( k + 1 ) } \ast \delta \omega = 0 \, ,
\end{align}
while
\begin{align}
    \delta \ast \omega = c_{n, k} ( - 1 )^{n ( k + 1 ) } \ast d \ast \ast \omega = c_{n, k} ( - 1 )^{n ( k + 1 ) } ( - 1 )^{k ( n - k) } \ast d \omega = 0 \, ,
\end{align}
so $\ast \omega$ is annihilated by both $d$ and $\delta$, and thus $\ast \omega \in \mathrm{Harm}^{n-k} ( \rhooneton )$. 

By a similar argument, we can construct a harmonic $k$-form from any harmonic $(n-k)$-form by taking its Hodge dual. The map $\ast : \mathrm{Harm}^k ( \rhooneton ) \leftrightarrow \mathrm{Harm}^{n-k} ( \rhooneton )$ is a bijection, since $\ast \ast$ is proportional to the identity. We conclude that
\begin{align}
    \mathrm{Harm}^k ( \rhooneton ) \cong \mathrm{Harm}^{n-k} ( \rhooneton ) \, ,
\end{align}
as desired.
\end{proof}

Finally, we conclude by pointing out a simple corollary of this result.

\begin{corollary}\label{symmetry_corollary}
    For any pure state density matrix $\rhooneton$, one has
    \begin{align}
        \dim \left( H^k ( \rhooneton ) \right) = \dim \left( H^{n-k} ( \rhooneton ) \right) \, .
    \end{align}
\end{corollary}

\begin{proof}

Combining the results of Theorems \ref{bijection} and \ref{hodge_duality} gives the chain of isomorphisms
\begin{align}
    H^k ( \rhooneton ) \cong \mathrm{Harm}^k ( \rhooneton ) \overset{\ast}{\cong} \mathrm{Harm}^{n-k} ( \rhooneton ) \cong H^{n-k} ( \rhooneton ) \, ,
\end{align}
where the symbol $\overset{\ast}{\cong}$ indicates that the isomorphism is supplied by the Hodge star map. Thus the dimensions of the $k$-th and $(n-k)$-th cohomology groups agree.
\end{proof}

The conclusion of Corollary \ref{symmetry_corollary} explains the symmetry property of the Poincar\'e polynomials which was pointed out around equation (\ref{example_poincare}).

\subsubsection*{\ul{\it Further remarks}}

In this work, we have been motivated by the set of analogies between de Rham cohomology and entanglement cohomology that are encoded in the following table.

\begin{center}
\begin{tabular}{||c | c ||} 
 \hline
 de Rham & Entanglement  \\ [0.5ex] 
 \hline\hline
 Manifold $\mathcal{M}$ & Pure state $\psioneton$ \\ 
 \hline
 Dimension $n$ & Number of subsystems $n$ in $\Honeton$ \\
 \hline
 $k$-forms $\omega_k \in \Omega^k$ & Restricted operators on Hilbert spaces $\mathcal{H}_{\ul{i}_1 \ldots \ul{i}_k}$ \\
 \hline
 $d$ built from antisymmetrized $\partial_\mu$ & $d$ built from antisymmetrized $\otimes \, \mathbb{I}$  \\
 \hline
 de Rham cohomology & entanglement cohomology \\ [1ex] 
 \hline
\end{tabular}
\end{center}

Our main results have been to extend this table to include several new rows, including the wedge product $\wedge$, the Hodge star $\ast$, the Hodge inner product $\langle \, \cdot \, , \, \cdot \, \rangle$, the codifferential $\delta$, and the Laplacian $\Delta$. All of these structures appear to behave in precisely the same way on both sides of the correspondence, naturally leading to results like the Hodge isomorphism and Hodge decomposition for entanglement cohomology, and the Hodge duality between $k$-forms and $(n-k)$-forms.

One might ask whether our results were guaranteed \emph{a priori} because of the basic structure of the entanglement cochain complex and the associated chain complex, which we have not discussed in this work, but which is treated in \cite{Mainiero:2019enr}. The answer to this question is negative, since as we mentioned above, one could have instead defined the so-called \texttt{GNS} complex using the alternate restriction (\ref{gns_restriction}). We are now prepared to understand why this definition would \emph{not} lead to a natural analogue of Hodge theory. If one had used this alternate restriction, entanglement $k$-forms would instead have been tuples of operators
\begin{align}\label{omega_other_restriction}
    \omega = \left( \omega_{\ul{I}_1} s_{\ul{I}_1} \, , \, \ldots \, , \, \omega_{\ul{I}_1} s_{\ul{I}_1}  \right) \, ,
\end{align}
with multiplication by support projections on the right but not on the left. Any element of such a tuple can ``absorb'' a multiplication by a support projection on the right while remaining unchanged, since $s_{\ul{I}_i}^2 = s_{\ul{I}_i}$, but is modified under a multiplication by a support projection on the left. However, as the definition of the inner product involves a trace, using the restriction $\big\|$ within the trace has the same effect as using $\big\vert$, since
\begin{align}
    \langle \omega, \eta \rangle &= \tr \left( \left( \omega \wedge \ast \eta \right) \big\|_{\rhooneton} \right) \nonumber \\
    &= \tr \left( \left( \omega \wedge \ast \eta \right) \soneton \right) \nonumber \\
    &= \tr \left( \left( \omega \wedge \ast \eta \right) \soneton \soneton \right) \nonumber \\
    &= \tr \left( \soneton \left( \omega \wedge \ast \eta \right) \soneton \right) \, ,
\end{align}
where we have used $\soneton^2 = \soneton$ and cyclicity of the trace. Therefore, for consistency of the definition of the inner product, the quantity $\omega \wedge \ast \eta $ must be able to ``absorb'' the action of a support projection from either side. But using operators constructed with the restriction (\ref{omega_other_restriction}), compatibility of supports only guarantees that the combination $\omega \wedge \ast \eta $ is unchanged under multiplication by support projections on the right, but not on the left. Therefore, the construction of an inner product fails for this alternate choice of restriction, and Hodge duality no longer holds. This explains why the Poincar\'e polynomials for the \texttt{GNS} complex, defined using the restriction $\big\|$, are not symmetric. This counterexample also demonstrates that the existence of our Hodge theory was not guaranteed simply due to the underlying cochain complex structure of the entanglement cohomology.

\section{Two-Qubit Examples}\label{sec:two_qubit_example}

The formalism we have reviewed and developed in Sections \ref{sec:ent_coho_review} and \ref{sec:hodge} is quite abstract, so it is useful to discuss a concrete example which illustrates the machinery of entanglement cohomology. Let us consider a bipartite Hilbert space $\mathcal{H} = \mathcal{H}_A \otimes \mathcal{H}_B$ where both $\mathcal{H}_A$ and $\mathcal{H}_B$ are 2-dimensional Hilbert spaces (qubits). We write\footnote{In Section \ref{sec:ent_coho_review} we labeled basis states using the integers from $1$ to $n$, rather than $0$ to $n - 1$, but here we adopt the zero-indexed convention to match the standard notation for the computational basis in QI.} the two basis states for $\mathcal{H}_A$ as
\begin{align}
    \ket{0_A} = \begin{bmatrix} 1 \\ 0 \end{bmatrix} \, , \qquad \ket{1_A} = \begin{bmatrix} 0 \\ 1 \end{bmatrix} \, , 
\end{align}
and likewise write $\ket{0_B}$ and $\ket{1_B}$ for the basis elements of $\mathcal{H}_B$. A basis for the tensor product $\mathcal{H}_{AB}$ is therefore formed by the four vectors
\begin{align}
    \ket{0_A 0_B} = \begin{bmatrix} 1 \\ 0 \\ 0 \\ 0 \end{bmatrix} \, , \quad \ket{0_A 1_B} = \begin{bmatrix} 0 \\ 1 \\ 0 \\ 0 \end{bmatrix} \, , \quad \ket{1_A 0_B} = \begin{bmatrix} 0 \\ 0 \\ 1 \\ 0 \end{bmatrix} \, , \quad \ket{1_A 1_B} = \begin{bmatrix} 0 \\ 0 \\ 0 \\ 1 \end{bmatrix} \, .
\end{align}
Let us construct and compare two states in $\mathcal{H}_{AB}$. The first will be written
\begin{align}
    \ket{\psi^{(P)}} = \ket{ 0_A 0_B } \, ,
\end{align}
where the ${}^{(P)}$ is short for ``product'' and where the associated density matrix is
\begin{align}\label{rhoF_defn}
    \rho^{(P)}_{AB} = \ket{\psi^{(P)}} \bra{\psi^{(P)}} = \begin{bmatrix} 1 & 0 & 0 & 0 \\ 0 & 0 & 0 & 0 \\ 0 & 0 & 0 & 0 \\ 0 & 0 & 0 & 0 \end{bmatrix} \, .
\end{align}
The second state we consider is
\begin{align}
    \ket{\psi^{(E)}} = \frac{1}{\sqrt{2}} \left( \ket{0_A 0_B} + \ket{ 1_A 1_B } \right) \, ,
\end{align}
where ${}^{(E)}$ is for ``entangled'' and whose corresponding density matrix ix
\begin{align}\label{rhoE_defn}
    \rho_{AB}^{(E)} = \ket{\psi^{(E)}} \bra{\psi^{(E)}} = \begin{bmatrix} \frac{1}{2} & 0 & 0 & \frac{1}{2} \\ 0 & 0 & 0 & 0 \\ 0 & 0 & 0 & 0 \\ \frac{1}{2} & 0 & 0 & \frac{1}{2} \end{bmatrix} \, .
\end{align}
Both of the density matrices (\ref{rhoF_defn}) and (\ref{rhoE_defn}) are projectors, because they are pure states, so the corresponding support projections are equal to the density matrices themselves:
\begin{align}
    s_{\rho_{AB}^{(P)}} = \rho_{AB}^{(P)} \, , \qquad s_{\rho_{AB}^{(E)}} = \rho_{AB}^{(E)} \, .
\end{align}
In this subsection we will also use the more compact notation
\begin{align}
    s_{AB}^{(P)} = s_{\rho_{AB}^{(P)}} \, , \qquad s_{AB}^{(E)} = s_{\rho_{AB}^{(E)}} \, ,
\end{align}
to avoid nested subscripts and superscripts.

First let us develop some intuition for the restriction maps $\big\vert_{\rho_{AB}^{(P)}}$ and $\big\vert_{\rho_{AB}^{(E)}}$ that project onto the images of these two density matrices. A general operator $\mathcal{O}_{AB}$ acting on $\mathcal{H}_{AB}$ can be expanded in a basis and its components may be represented as a $4 \times 4$ matrix:
\begin{align}
    \mathcal{O}_{AB} = \begin{bmatrix} \mathcal{O}_{11} & \mathcal{O}_{12} & \mathcal{O}_{13} & \mathcal{O}_{14} \\ \mathcal{O}_{21} & \mathcal{O}_{22} & \mathcal{O}_{23} & \mathcal{O}_{24} \\ \mathcal{O}_{31} & \mathcal{O}_{32} & \mathcal{O}_{33} & \mathcal{O}_{34} \\ \mathcal{O}_{41} & \mathcal{O}_{42} & \mathcal{O}_{43} & \mathcal{O}_{44} \end{bmatrix} \, .
\end{align}
The restrictions of such a general operator are
\begin{align}\label{F_restriction}
    \mathcal{O}_{AB} \big\vert_{\rho_{AB}^{(P)}} &= s_{AB}^{(P)} \mathcal{O}_{AB} s_{AB}^{(P)} \nonumber \\
    &= \mathcal{O}_{11} \begin{bmatrix} 1 & 0 & 0 & 0 \\ 0 & 0 & 0 & 0 \\ 0 & 0 & 0 & 0 \\ 0 & 0 & 0 & 0 \end{bmatrix} \, ,
\end{align}
and
\begin{align}\label{E_restriction}
    \mathcal{O}_{AB} \big\vert_{\rho_{AB}^{(E)}} &= s_{AB}^{(E)} \mathcal{O}_{AB} s_{AB}^{(E)} \nonumber \\
    &= \frac{1}{4} \left( \mathcal{O}_{11} + \mathcal{O}_{14} + \mathcal{O}_{41} + \mathcal{O}_{44} \right) \begin{bmatrix} 1 & 0 & 0 & 1 \\ 0 & 0 & 0 & 0 \\ 0 & 0 & 0 & 0 \\ 1 & 0 & 0 & 1 \end{bmatrix} \, ,
\end{align}
We therefore see that both of the restrictions (\ref{F_restriction}) and (\ref{E_restriction}) are sensitive to only a single linear combination of the matrix elements of $\mathcal{O}_{AB}$. This is to be expected, since again both density matrices are rank-one projectors, and thus their images are one-dimensional.

In order to study the entanglement complex, we will also need to consider the reduced density matrices obtained by tracing out one of the subsystems. These are
\begin{align}
    \rho_A^{(P)} &= \tr_B \left( \rho_{AB}^{(P)} \right) = \begin{bmatrix} 1 & 0 \\ 0 & 0 \end{bmatrix} = \tr_A \left( \rho_{AB}^{(P)} \right) = \rho_B^{(P)} \, , \nonumber \\
    \rho_A^{(E)} &= \tr_B \left( \rho_{AB}^{(E)} \right) = \frac{1}{2} \begin{bmatrix} 1 & 0 \\ 0 & 1 \end{bmatrix} = \tr_A \left( \rho_{AB}^{(E)} \right) = \rho_B^{(E)} \, .
\end{align}
The associated support projections will now have different properties, since the reduced density matrix associated with the entangled state is full-rank (in fact it is proportional to the identity), whereas $\rho_A^{(P)} = \rho_B^{(P)}$ is still rank $1$:
\begin{align}
    s_{\rho_A^{(P)}} = \begin{bmatrix} 1 & 0 \\ 0 & 0 \end{bmatrix} = s_{\rho_B^{(P)}} \, , \nonumber \\
    s_{\rho_A^{(E)}} = \begin{bmatrix} 1 & 0 \\ 0 & 1 \end{bmatrix} = s_{\rho_B^{(E)}} \, .
\end{align}
As above, we will abbreviate these support projections using the compact notation $s_A^{(X)}$ for $s_{\rho_A^{(X)}}$ and $s_B^{(X)}$ for $s_{\rho_B^{(X)}}$, where $X$ is either $E$ or $P$.

If we again take a generic operator
\begin{align}
    \mathfrak{o}_A = \begin{bmatrix} \mathfrak{o}_{11} & \mathfrak{o}_{12} \\ \mathfrak{o}_{21} & \mathfrak{o}_{22} \end{bmatrix} \, ,
\end{align}
then
\begin{align}
    \mathfrak{o}_A \big\vert_{\rho_A^{(P)}} &= s_A^{(P)} \mathfrak{o} s_A^{(P)} \nonumber \\
    &= \begin{bmatrix} \mathfrak{o}_{11} & 0 \\ 0 & 0 \end{bmatrix} \, ,
\end{align}
while
\begin{align}
    \mathfrak{o}_A \big\vert_{\rho_A^{(E)}} &= s_A^{(E)} \mathfrak{o} s_A^{(E)} \nonumber \\
    &= \begin{bmatrix} \mathfrak{o}_{11} & \mathfrak{o}_{12} \\ \mathfrak{o}_{21} & \mathfrak{o}_{22} \end{bmatrix} \nonumber \\
    &= \mathfrak{o}_A \, .
\end{align}
Let us now consider the ingredients used to build the commutant complexes associated with our two states, which take the schematic form
\begin{align}
    0 \to \mathbb{C} \xrightarrow{d^{(X)}} \Omega^1 ( \rho_{AB}^{(X)} ) \xrightarrow{d^{(X)}} \Omega^2 ( \rho_{AB}^{(X)} ) \xrightarrow{d^{(X)}} 0 \, ,
\end{align}
where again ${}^{(X)}$ is used as a placeholder for either ${}^{(P)}$ or ${}^{(E)}$. We use the same symbol $d^{(X)}$ for the coboundary operators $d^{(X), 0}$, $d^{(X), 1}$, $d^{(X), 2}$, distinguishing between them based on context. Recall that an element of $\Omega^2 ( \rho_{AB}^{(X)} )$ is simply a $4 \times 4$ matrix which has been restricted by left- and right-multiplying by the appropriate support projection:
\begin{align}
    \omega_2 \in \Omega^2 ( \rho_{AB}^{(X)} ) \; \implies \; \omega_2 = s_{AB}^{(X)} \mathcal{O}_{AB} s_{AB}^{(X)} \, \text{ for some } \mathcal{O}_{AB} \, .
\end{align}
That is, any entanglement $2$-form is a $4 \times 4$ matrix of either the form (\ref{F_restriction}) or (\ref{E_restriction}), depending on whether we are studying $\rho_{AB}^{(P)}$ or $\rho_{AB}^{(E)}$. Thus in either case $\Omega^2 ( \rho_{AB}^{(X)} )$ is a one-dimensional vector space.

On the other hand, an element $\omega_1 \in \Omega^1 \left( \rho_{AB}^{(X)} \right)$ is a tuple of $2 \times 2$ matrices:
\begin{align}
    \omega_1 \in \Omega^1 ( \rho_{AB}^{(X)} ) \; \implies \; \omega_1 = \left( s_{A}^{(X)} \mathcal{O}_{A} s_{A}^{(X)} \, , \, s_{B}^{(X)} \mathcal{O}_{B} s_{B}^{(X)} \right) \, \text{ for some } \mathcal{O}_{A} \, , \mathcal{O}_B \, .
\end{align}
For instance, in the case of $\rho_{AB}^{(P)}$, any entanglement $1$-form is a tuple
\begin{align}\label{factorized_one_form}
    \omega_1 \in \Omega^1 ( \rho_{AB}^{(P)} ) \; \implies \; \omega_1 = \left( \begin{bmatrix} a & 0 \\ 0 & 0 \end{bmatrix} \, , \, \begin{bmatrix} b & 0 \\ 0 & 0 \end{bmatrix} \right) \, ,
\end{align}
for some constants $a$ and $b$, while for $\rho_{AB}^{(E)}$, 
\begin{align}\label{general_entangled_one_form}
    \omega_1 \in \Omega^1 ( \rho_{AB}^{(E)} ) \; \implies \; \omega_1 = \left( \begin{bmatrix} a_{11} & a_{12} \\ a_{21} & a_{22} \end{bmatrix} \, , \, \begin{bmatrix} b_{11} & b_{12} \\ b_{21} & b_{22} \end{bmatrix} \right) \, ,
\end{align}
for some $a_{ij}$ and $b_{ij}$. Thus we see that $\Omega^1 ( \rho_{AB}^{(P)} )$ has dimension $2$ while $\dim \left( \Omega^1 ( \rho_{AB}^{(E)} ) \right) = 8$.

Next consider how the coboundary operators $d^{(X)}$ act in these complexes. The first operator, which acts on complex numbers $\lambda \in \mathbb{C}$, simply sends
\begin{align}
    d^{(X)} : \lambda \to ( \lambda s_A^{(X)} , \lambda s_B^{(X)} ) \, ,
\end{align}
so that for the product state one has
\begin{align}\label{exact_factorized}
    d^{(P)} \lambda = \left( \begin{bmatrix} \lambda & 0 \\ 0 & 0 \end{bmatrix} \, , \, \begin{bmatrix} \lambda & 0 \\ 0 & 0 \end{bmatrix} \right) \, ,
\end{align}
while for the entangled state we see
\begin{align}\label{exact_entangled}
    d^{(E)} \lambda = \left( \begin{bmatrix} \lambda & 0 \\ 0 & \lambda \end{bmatrix} \, , \, \begin{bmatrix} \lambda & 0 \\ 0 & \lambda \end{bmatrix} \right) \, .
\end{align}
Therefore, in both complexes there is a one-dimensional space of exact entanglement one-forms. In fact, this feature is generic, since in \emph{any} entanglement complex the image of the first coboundary operator acting on $\mathbb{C}$ is necessarily one-dimensional. Speaking loosely, one might say that entanglement one-forms of the type (\ref{exact_factorized}) and (\ref{exact_entangled}) are ``pure gauge'' in the complexes associated with the product and entangled states, respectively.

Next let us consider the second coboundary operator, $d : \Omega^1 \to \Omega^2$, which maps entanglement one-forms to entanglement two-forms as
\begin{align}
    d^{(X)} : \left( \mathcal{O}_A , \mathcal{O}_B \right) \to \left( \mathbb{I}_A \otimes \mathcal{O}_B - \mathcal{O}_A \otimes \mathbb{I}_B \right) \big\vert_{\rho_{AB}^{(X)}} \, .
\end{align}
Note that, by the definition of $\Omega^1$, the input forms $\mathcal{O}_A$ and $\mathcal{O}_B$ have already been restricted so that $\mathcal{O}_A = \mathcal{O}_A \big\vert_{\rho_A^{(X)}}$ and $\mathcal{O}_B = \mathcal{O}_B \big\vert_{\rho_B^{(X)}}$, so we do not write the bars explicitly.

It is straightforward to check that, for the product state,
\begin{align}
    d^{(P)} \left( \begin{bmatrix} a & 0 \\ 0 & 0 \end{bmatrix} \, , \,  \begin{bmatrix} b & 0 \\ 0 & 0 \end{bmatrix} \right) = \left( b - a \right) \cdot \begin{bmatrix} 1 & 0 & 0 & 0 \\ 0 & 0 & 0 & 0 \\ 0 & 0 & 0 & 0 \\ 0 & 0 & 0 & 0 \end{bmatrix} \, ,
\end{align}
which means that an entanglement one-form (\ref{factorized_one_form}) in the $\rho^{(P)}_{AB}$ complex is closed if and only if $a = b$. But for a one-form with $a = b$, we see from equation (\ref{exact_factorized}) that this form is also exact. This means that the kernel of $d^{(P)} : \Omega^1 \left( \rho_{AB}^{(P)} \right) \to \Omega^2 \left( \rho_{AB}^{(P)} \right)$ and the image of $d^{(P)} : \mathbb{C} \to \Omega^1 \left( \rho_{AB}^{(P)} \right)$ coincide, and thus for the product state, 
\begin{align}
    H^1 \left( \rho_{AB}^{(P)} \right) = \frac{\ker \left( d^{(P)} : \Omega^1 \left( \rho_{AB}^{(P)} \right) \to \Omega^2 \left( \rho_{AB}^{(P)} \right) \right)}{\mathrm{im} \left( d^{(P)} : \mathbb{C} \to \Omega^1 \left( \rho_{AB}^{(P)} \right) \right) } = \left\{ 0 \right\} \, ,
\end{align}
so the cohomology group $H^1 \left( \rho_{AB}^{(P)} \right)$ is trivial. This is just as we would expect, since the cohomology of the entanglement complex is (of course) supposed to measure entanglement, and the product state $\rho_{AB}^{(P)}$ has no entanglement by definition.

We may repeat this exercise for the entangled state. When acting on a generic entanglement one-form (\ref{general_entangled_one_form}), the output of the $d^{(E)}$ operation is
\begin{align}
    d^{(E)} \left( \begin{bmatrix} a_{11} & a_{12} \\ a_{21} & a_{22} \end{bmatrix} \, , \, \begin{bmatrix} b_{11} & b_{12} \\ b_{21} & b_{22} \end{bmatrix} \right) = \left( \frac{b_{11} - a_{11}}{4} + \frac{b_{22} - a_{22}}{4} \right) \cdot \begin{bmatrix} 1 & 0 & 0 & 1 \\ 0 & 0 & 0 & 0 \\ 0 & 0 & 0 & 0 \\ 1 & 0 & 0 & 1 \end{bmatrix} \, ,
\end{align}
which vanishes if and only if
\begin{align}\label{vanish_condition}
    \frac{b_{11} - a_{11}}{4} + \frac{b_{22} - a_{22}}{4} = 0 \, .
\end{align}
Equation (\ref{vanish_condition}) is one linear condition, which fixes one of the $8$ free parameters that determine an entanglement one-form associated with the entangled state in terms of the other parameters. This leaves a $7$-dimensional space of closed $1$-forms, which is to be modded out by the $1$-dimensional space of exact $1$-forms. We therefore expect the dimension of the cohomology to be
\begin{align}
    \dim \left( H^1 \left( \rho_{AB}^{(E)} \right) \right) = \dim \left( \frac{\ker \left( d^{(E)} : \Omega^1 \left( \rho_{AB}^{(E)} \right) \to \Omega^2 \left( \rho_{AB}^{(E)} \right) \right)}{\mathrm{im} \left( d^{(E)} : \mathbb{C} \to \Omega^1 \left( \rho_{AB}^{(E)} \right) \right) } \right) = 6 \, .
\end{align}
Indeed, it is not difficult to find representatives of these six cohomology classes. One can check that the entanglement one-forms
\begin{gather}
    \left( \begin{bmatrix} 0 & 0 \\ 0 & 0 \end{bmatrix} \, , \,  \begin{bmatrix} 1 & 0 \\ 0 & -1 \end{bmatrix} \right) \, , \quad \left( \begin{bmatrix} 1 & 0 \\ 0 & -1 \end{bmatrix} \, , \, \begin{bmatrix} 0 & 0 \\ 0 & 0 \end{bmatrix} \right) \, , \nonumber \\
    \left( \begin{bmatrix} 0 & 0 \\ 0 & 0 \end{bmatrix} \, , \,  \begin{bmatrix} 0 & 0 \\ 1 & 0 \end{bmatrix} \right) \, , \quad \left( \begin{bmatrix} 0 & 0 \\ 1 & 0 \end{bmatrix} \, , \, \begin{bmatrix} 0 & 0 \\ 0 & 0 \end{bmatrix} \right) \, , \nonumber \\
    \left( \begin{bmatrix} 0 & 0 \\ 0 & 0 \end{bmatrix} \, , \,  \begin{bmatrix} 0 & 1 \\ 0 & 0 \end{bmatrix} \right) \, , \quad \left( \begin{bmatrix} 0 & 1 \\ 0 & 0 \end{bmatrix} \, , \, \begin{bmatrix} 0 & 0 \\ 0 & 0 \end{bmatrix} \right) \, , \label{representatives}
\end{gather}
are all closed under the $d^{(E)}$ operator and are ``gauge-inequivalent'' in the sense that no pair of the forms (\ref{representatives}) differ by an exact entanglement $1$-form.

Finally, let us discuss some Hodge-theoretic aspects of these examples. We begin with the product state. Given a generic $\omega_1 \in \Omega^1 ( \rho_{AB}^{(P)} )$, its Hodge dual is
\begin{align}
    \omega_1 = \left( \begin{bmatrix} a & 0 \\ 0 & 0 \end{bmatrix} \, , \, \begin{bmatrix} b & 0 \\ 0 & 0 \end{bmatrix} \right) \; \implies \; \ast \omega_1 = \left( \begin{bmatrix} - b^\ast & 0 \\ 0 & 0 \end{bmatrix} \, , \, \begin{bmatrix} a^\ast & 0 \\ 0 & 0 \end{bmatrix} \right) \, ,
\end{align}
where $a^\ast$, $b^\ast$ are the complex conjugates of $a$ and $b$. The norm induced by the Hodge inner product is simply
\begin{align}
    \langle \omega_1, \omega_1 \rangle = | a |^2 + | b |^2 \, ,
\end{align}
which is manifestly positive-definite.

More generally, given a second one-form $\eta_1 \in \Omega^1 ( \rho_{AB}^{(P)} )$,
\begin{align}
    \eta_1 = \left( \begin{bmatrix} c & 0 \\ 0 & 0 \end{bmatrix} \, , \, \begin{bmatrix} d & 0 \\ 0 & 0 \end{bmatrix} \right) \, , 
\end{align}
one finds
\begin{align}
    \langle \omega_1, \eta_1 \rangle = a c^\ast + b d^\ast \, .
\end{align}
Let us now see the analogous expressions for the entangled state $\rho_{AB}^{(E)}$. Fix two one-forms $\omega_1, \eta_1 \in \Omega^1 ( \rho_{AB}^{(E)} )$ with expansions
\begin{align}
    \omega_1 = \left( \begin{bmatrix} a_{11} & a_{12} \\ a_{21} & a_{22} \end{bmatrix} \, , \, \begin{bmatrix} b_{11} & b_{12} \\ b_{21} & b_{22} \end{bmatrix} \right) \, , \qquad \eta_1 = \left( \begin{bmatrix} c_{11} & c_{12} \\ c_{21} & c_{22} \end{bmatrix} \, , \, \begin{bmatrix} d_{11} & d_{12} \\ d_{21} & d_{22} \end{bmatrix} \right) \, .
\end{align}
The Hodge star in the $\rho_{AB}^{(E)}$ complex acts as
\begin{align}
    \ast \omega_1 = \left( \begin{bmatrix} - b_{11}^\ast & - b_{12}^\ast \\ - b_{21}^\ast & -b_{22}^\ast \end{bmatrix} \, , \, \begin{bmatrix} a_{11}^\ast & a_{12}^\ast \\ a_{21}^\ast & a_{22}^\ast \end{bmatrix} \right) \, .
\end{align}
The inner product is
\begin{align}
    \langle \omega_1 , \eta_1 \rangle = \frac{1}{2} \left( \begin{bmatrix} a_{11} & a_{12} \\ a_{21} & a_{22} \end{bmatrix} \cdot \begin{bmatrix} c_{11} & c_{12} \\ c_{21} & c_{22} \end{bmatrix}^\dagger + \begin{bmatrix} b_{11} & b_{12} \\ b_{21} & b_{22} \end{bmatrix} \cdot \begin{bmatrix} d_{11} & d_{12} \\ d_{21} & d_{22} \end{bmatrix}^\dagger \right) \, ,
\end{align}
and in particular we see that $\langle \omega_1, \omega_1 \rangle = \frac{1}{2} \tr \left( A A^\dagger + B B^\dagger \right)$ where $A = \begin{bmatrix} a_{11} & a_{12} \\ a_{21} & a_{22} \end{bmatrix}$, $B = \begin{bmatrix} b_{11} & b_{12} \\ b_{21} & b_{22} \end{bmatrix}$, which is again positive-definite as expected. This form of the inner product for the $\rho_{AB}^{(E)}$ complex is a consequence of the fact that the two Schmidt coefficients $\lambda_\alpha$ for this state are equal, so the inner product collapses as in equation (\ref{inner_product_same_lambdas}).

\section{Conclusion}\label{sec:conclusion}

In this work, we have explored the use of homological tools to understand entanglement in finite-dimensional quantum systems, extending the analysis of \cite{Mainiero:2019enr}. After reviewing the construction of a cochain complex associated with a generic pure state in any finite-dimensional, multi-partite Hilbert space, we have developed a Hodge theory for this cochain complex. In particular, we defined notions of inner product, codifferential, and Laplacian on entanglement $k$-forms, and proved analogues of the Hodge isomorphism theorem and Hodge decomposition for entanglement cohomology. To do this, we constructed a Hodge star operator which maps entanglement $k$-forms to entanglement $(n-k)$ forms, in a way which sends harmonic forms to harmonic forms, which proves that the dimensions of entanglement cohomologies enjoy a symmetry property. These observations identify and explain new patterns in the mathematical structure of the entanglement complex, which may be useful for understanding types of multi-partite entanglement.

There remain several interesting directions for future inquiry. One of the most obvious is to investigate whether an analogue of the machinery of entanglement cohomology -- and its Hodge-theoretic extension considered here -- applies in quantum field theory, where the Hilbert space is infinite-dimensional and does not admit a conventional tensor product structure. Some initial comments about this generalization already appeared in \cite{Mainiero:2019enr}.

Let us outline a few other future directions below.

\subsubsection*{\ul{\it Mixed states and reflected cohomology}}

Our discussion has focused on pure states and their entanglement. In the original work \cite{Mainiero:2019enr}, entanglement cohomology was also applied to mixed states, where it was shown that cohomological data is related to a rather weak condition dubbed ``support factorizability'' in such mixed states. An important future direction is to investigate whether more fine-grained information about mixed states can also be extracted using homological tools.

One possible strategy for doing this is ``going to the church of the larger Hilbert space'' in the sense that any mixed state involving a finite number of subsystems is equivalent, by the process of purification, to a pure state in a multipartite system with a larger number of subsystems. For instance, one could envision taking any mixed state described by a density matrix $\rho$, constructing its canonical purification $\ket{\sqrt{\rho}}$, and then assembling the dimensions of cohomologies for both this pure states and all of its reduced subsystems. It would be interesting to investigate whether this procedure provides enough information to give some classification of possible entanglement structures in mixed states.

Let us note that, in the case of a mixed state in a bipartite system $\mathcal{H}_{AB}$, the purification procedure yields a pure state in a larger Hilbert space $\mathcal{H}_{A A^\ast B B^\ast}$. The collection of all dimensions of cohomologies for reduced systems therefore contains data about entanglement in subsystems like $\mathcal{H}_{A A^\ast}$ and $\mathcal{H}_{B B^\ast}$. By analogy with reflected entropy \cite{Dutta:2019gen}, it seems natural to refer to this structure as ``reflected cohomology.'' An exciting future direction is to study the properties of such reflected cohomologies in general mixed states.

\subsubsection*{\ul{\it Connection to ``generative effects''}}

The motivation for the approach taken in this work is that cohomology gives a natural language for discussing an obstruction from lifting local properties to global properties, such as realizing a state in a ``global'' tensor product Hilbert space as a tensor product of states in ``local'' subsystem Hilbert spaces. Said differently, cohomology gives a mechanism for characterizing the extent to which ``the whole is greater than the sum of its parts'' in the sense that additional phenomena, such as entangled states, emerge in a composite system despite being absent in any of the component subsystems.

Another framework for analyzing qualitatively similar phenomena is that of \emph{generative effects}, which were introduced in \cite{Adam2017SystemsGA} and are nicely reviewed in Chapter 1 of \cite{fong2018sevensketchescompositionalityinvitation}. Such generative effects can be defined in quite general categories, but for our purposes, it suffices to restrict to the setting of preorders, which are sets equipped with a comparison operation $\leq$ that is symmetric and transitive.  Recall that the \emph{join} of a collection of elements in a preorder is, roughly speaking, their least upper bound, and the \emph{meet} of a collection of elements is roughly their greatest lower bound (the join and meet precisely coincide with the notions of supremum and infimum if the preorder is also a \emph{total order}). Furthermore, a monotone map between preorders is a function with the property that, if $x \leq y$, then $f ( x ) \leq f ( y )$. We say that a monotone map $f$ has a \emph{generative effect} if it does not preserve joins, that is, if there exists at least one pair of elements $a$, $b$ such that
\begin{align}
    f ( a ) \vee f ( b ) \not\cong f ( a \vee b ) \, ,
\end{align}
where we say $x \cong y$ if $x \leq y$ and $y \leq x$ in the preorder.

Suppose that we interpret such a function as a measurement or observation of some collection of systems. Then a function with generative effects, morally speaking, exhibits additional structure when applied to composite systems which is not captured by combining observations applied to individual subsystems. Such a scenario models new effects that occur due to the interconnections between subsystems, much like the existence of entangled states in tensor product Hilbert spaces.

This schematic connection between generative effects and entanglement can likely be made precise in several different ways. An interesting direction for future research is to see whether this construction, or other ways of presenting entanglement as generative effects, also lead to useful classification schemes or physical insights. The original work \cite{Adam2017SystemsGA} presented a quite general framework for building cohomologies associated with generative effects, and it may be that one such cohomology coincides with the notion of entanglement cohomology considered here and in \cite{Mainiero:2019enr}. Furthermore, a more sophisticated version of this construction (perhaps a functor between categories which does not preserve colimits, rather than a monotone map that does not preserve joins) might capture even richer information about entanglement.

\subsubsection*{\ul{\it Seeking structure through machine learning}}

The focus of this work has been on the mathematical structure of entanglement cohomology rather than on a systematic numerical investigation, although we have implemented routines for various operations on entanglement complexes in the Python programming language, building on the QuTiP library \cite{Johansson:2011jer,Johansson:2012qtx}.\footnote{Our Python library will be made publicly available at a future time.} An advantage of having access to such a Python implementation is that it facilitates interfacing with standard libraries for data science and machine learning, including scikit-learn, NumPy, SciPy, and others.

Just as we have discovered and explained one particular pattern in the dimensions of entanglement cohomologies in this work -- namely, the symmetry property of the Poincar\'e polynomials -- one might hope that a data-driven exploratory analysis might reveal still other patterns and structures that might likewise be explained mathematically. To this end, one might use machine learning and artificial intelligence techniques for conjecture generation, a strategy which has been successfully applied to several other problems; see \cite{Gukov:2024buj} and references therein for an introduction.

\section*{Acknowledgements}

We thank Ning Bao, Christopher Beasley, Keiichiro Furuya, James Halverson, Sarah Harrison, Veronika Hubeny, Ziming Ji, Frederic Jia, Yikun Jiang, Joydeep Naskar, Mukund Rangamani, Fabian Ruehle, Benjamin Sung, and Julio Virrueta for helpful comments and discussions related to this work. We are also grateful to Tom Mainiero, Colin Nancarrow, and Eugene Tang for feedback on a draft of this manuscript.
C.\,F. is supported by DE-SC0009999, funds from the University of California, and the National Science Foundation under Cooperative Agreement PHY-2019786 (the NSF AI Institute for Artificial Intelligence and Fundamental Interactions).

\appendix

\section{Compatibility of Supports}\label{app:compatibility}

In this Appendix, we provide a simple proof of the compatibility condition (\ref{compatibility_bipartite}) in the case of a pure state on a bipartite system.

Given a state $\ket{\psi_{AB}} \in \mathcal{H}_{AB} = \mathcal{H}_A \otimes \mathcal{H}_B$, we perform a Schmidt decomposition
\begin{align}
    \ket{ \psi_{AB} } = \sum_{\alpha = 1}^{S} \lambda_\alpha \ket{ \alpha_{A} } \otimes \ket{ \alpha_{B} } \, ,
\end{align}
where $S$ is the Schmidt rank of $\ket{ \psi_{AB} }$. The associated density matrix is
\begin{align}
    \rho_{AB} &= \ket{\psi_{AB}} \bra{\psi_{AB}} \nonumber \\
    &= \sum_{\alpha, \beta = 1}^{S} \lambda_\alpha \lambda_\beta \left( \ket{ \alpha_A } \otimes \ket{\alpha_B } \right) \left( \bra{\beta_A} \otimes \bra{\beta_B} \right) \, ,
\end{align}
and since $\rho_{AB}$ is manifestly a rank-$1$ projection operator, we have $\rho_{AB} = s_{AB}$.

The reduced density matrices are
\begin{align}
    \rho_A = \sum_{\alpha = 1}^{S} \lambda_\alpha^2 \ket{\alpha_A} \bra{\alpha_A} \, , \qquad  \rho_B = \sum_{\alpha = 1}^{S} \lambda_\alpha^2 \ket{\alpha_B} \bra{\alpha_B} \, .
\end{align}
As the bases $\ket{\alpha_A}$ and $\ket{\alpha_B}$ are orthonormal, the corresponding support projections are
\begin{align}
    s_A = \sum_{\alpha = 1}^{S} \ket{\alpha_A} \bra{\alpha_A} \, , \qquad s_B = \sum_{\alpha = 1}^{S} \ket{\alpha_B} \bra{\alpha_B} \, .
\end{align}
Now consider the combination
\begin{align}
    &\left( s_A \otimes s_B \right) s_{AB} \nonumber \\
    &\quad = \left( \left( \sum_{\alpha = 1}^{S} \ket{\alpha_A} \bra{\alpha_A} \right) \otimes \left( \sum_{\beta = 1}^{S} \ket{\beta_B} \bra{\beta_B} \right) \right) \left( \sum_{\gamma, \delta = 1}^{S} \lambda_\gamma \lambda_\delta \left( \ket{ \gamma_A } \otimes \ket{\gamma_B } \right) \left( \bra{\delta_A} \otimes \bra{\delta_B} \right)  \right) \nonumber \\
    &\quad = \sum_{\alpha, \beta, \gamma, \delta = 1}^{S} \lambda_\gamma \lambda_\delta \left( \ket{\alpha_A} \braket{\alpha_A}{\gamma_A} \bra{\delta_A} \right) \otimes \left( \ket{\beta_B} \braket{\beta_B}{\gamma_B} \bra{\delta_B} \right) \nonumber \\
    &\quad = \sum_{\alpha, \beta, \gamma, \delta = 1}^{S} \lambda_\gamma \lambda_\delta \left( \ket{\alpha_A} \delta_{\alpha \gamma} \bra{\delta_A} \right) \otimes \left( \ket{\beta_B} \delta_{\beta \gamma} \bra{\delta_B} \right) \nonumber \\
    &\quad = \sum_{\gamma, \delta = 1}^{S} \lambda_\gamma \lambda_\delta \ket{\gamma_A} \bra{\delta_A} \otimes \ket{\gamma_B} \bra{\delta_B} \nonumber \\
    &\quad = s_{AB} \, .
\end{align}
Here the symbols $\delta_{\alpha \gamma}$ and $\delta_{\beta \gamma}$ represent the Kronecker delta and are not to be confused with the index $\delta$. By an almost identical sequence of steps, one has
\begin{align}
    &s_{AB} \left( s_A \otimes s_B \right)  \nonumber \\
    &\quad = \left( \sum_{\gamma, \delta = 1}^{S} \lambda_\gamma \lambda_\delta \left( \ket{ \gamma_A } \otimes \ket{\gamma_B } \right) \left( \bra{\delta_A} \otimes \bra{\delta_B} \right)  \right)  \left( \left( \sum_{\alpha = 1}^{S} \ket{\alpha_A} \bra{\alpha_A} \right) \otimes \left( \sum_{\beta = 1}^{S} \ket{\beta_B} \bra{\beta_B} \right) \right) \nonumber \\
    &\quad = \sum_{\alpha, \beta, \gamma, \delta = 1}^{S} \lambda_\gamma \lambda_\delta \left( \ket{\gamma_A} \braket{\delta_A}{\alpha_A} \bra{\alpha_A} \right) \otimes \left( \ket{\gamma_B} \braket{\delta_B}{\beta_B} \bra{\beta_B} \right)\nonumber \\
    &\quad = \sum_{\alpha, \beta, \gamma, \delta = 1}^{S} \lambda_\gamma \lambda_\delta \left( \ket{\gamma_A} \delta_{\delta \alpha} \bra{\alpha_A} \right) \otimes \left( \ket{\gamma_B} \delta_{\delta \beta} \bra{\beta_B} \right)\nonumber \\
    &\quad = \sum_{\gamma, \delta = 1}^{S} \lambda_\gamma \lambda_\delta \left( \ket{\gamma_A} \bra{\delta_A} \right) \otimes \left( \ket{\gamma_B} \bra{\delta_B} \right)\nonumber \\
    &\quad = s_{AB} \, .
\end{align}
We conclude that
\begin{align}
   \left( s_A \otimes s_B \right) s_{AB} = s_{AB} =  s_{AB} \left( s_A \otimes s_B \right) \, ,
\end{align}
which is what we set out to show.

\bibliographystyle{utphys}
\bibliography{master}

\end{document}